\documentclass[10 pt,twocolumn]{IEEEtran}

\usepackage[dvips]{graphicx}
\usepackage{amssymb}
\usepackage{graphicx}
\usepackage[cmex10]{amsmath}
\usepackage{algorithm}
\usepackage{algorithmic}
\usepackage{array}
\usepackage{eqparbox}
\usepackage[tight,footnotesize]{subfigure}
\usepackage{float}
\usepackage{url}
\usepackage{paralist}
\usepackage{multirow}
\usepackage{wrapfig}

\newtheorem{thm}{Theorem}[section]
\newtheorem{lem}{Lemma}

\begin{document}

\title{Target Tracking via Crowdsourcing: A Mechanism Design Approach}


\author{\IEEEauthorblockN{Nianxia Cao, Swastik Brahma, Pramod K. Varshney\\}
\IEEEauthorblockA{Department of Electrical Engineering and Computer Science, Syracuse University, NY, 13244, USA\\
\{ncao, skbrahma, varshney\}@syr.edu}}

\maketitle


\begin{abstract}
\boldmath In this paper, we propose a crowdsourcing based framework for myopic target tracking by designing an incentive-compatible mechanism based optimal auction in a wireless sensor network (WSN) containing sensors that are selfish and profit-motivated. For typical WSNs which have limited bandwidth, the fusion center (FC) has to distribute the total number of bits that can be transmitted from the sensors to the FC among the sensors. To accomplish the task, the FC conducts an auction by soliciting bids from the selfish sensors, which reflect how much they value their energy cost. Furthermore, the rationality and truthfulness of the sensors are guaranteed in our model. The final problem is formulated as a multiple-choice knapsack problem (MCKP), which is solved by the dynamic programming method in pseudo-polynomial time. Simulation results show the effectiveness of our proposed approach in terms of both the tracking performance and lifetime of the sensor network.
\end{abstract}


%
\IEEEpeerreviewmaketitle

\section{Introduction}
Crowdsourcing is the practice of obtaining needed services, ideas, or content by soliciting contributions from a large group of people, rather than from traditional employees or suppliers \cite{wiki:crowdsourcing}. Many of today's sensing applications allow users carrying devices with built-in sensors,
such as sensors built in smart phones, and automobiles, to contribute towards
an inference task with their sensing measurements, which is exactly an application of crowdsourcing. 
For instance, today's smart phones are embedded with various sensors, such as camera, microphone,
accelerometer, and GPS, which can be used to acquire information regarding a phenomenon of interest. 
An advantage of such architectures is that they do not need a dedicated sensing infrastructure for different inference tasks,
thereby providing cost effectiveness.
Another advantage of such architectures is that they allow ubiquitous coverage.

Systems and applications that rely on utilizing an infrastructure where crowdsourced sensing measurements of
participating users are used
are poised to revolutionize many sectors of our life.
Some example application domains include 
social networks, environmental monitoring~\cite{Envir_moni,noisePol},
green computing~\cite{ganti2010greengps},
target localization and tracking~\cite{bikeTrack,PositionIt,mobileObjectAlgo,Nest,ObjectsCallingHome},
healthcare~\cite{reddy2007image} (such as predicting and tracking disease patterns/outbreaks),
and tracking traffic patterns \cite{yoon2007surface,kerner2005traffic}.
For instance, the OpenSense project~\cite{Envir_moni} involves the design of a 
sensing infrastructure for real-time air quality monitoring using heterogeneous
sensors owned by the general public, while~\cite{noisePol} involves the
design of a similar system to monitor noise levels.
GreenGPS~\cite{ganti2010greengps} uses data from sensors installed in automobiles to map fuel
consumption on city streets and construct fuel efficient routes between arbitrary end-points.
Various systems to estimate object locations and to track them using smartphone sensors
have also been proposed. For instance, work reported in \cite{PositionIt,Nest} utilizes
built-in sensors in smartphones such as camera, digital compass and GPS,
to estimate a target location as well as monitor the velocity of moving objects. In \cite{bikeTrack,mobileObjectAlgo,ObjectsCallingHome},
proximity sensors in built-in smartphones are used 
to track objects (such as lost/stolen devices) installed with electronic tags
(such as Bluetooth or RFID tags).
Such systems have important commercial applications (such as tracking lost/stolen objects or accurately estimating arrival time of buses)
as well as defense related applications (estimating the enemy's vehicle position prior to an attack).


Existing sensing applications and systems assume voluntary participation of users,
for example,~\cite{paticipatory_sensing_Estrin,yoon2007surface,bikeTrack,
PositionIt,mobileObjectAlgo,Nest,ObjectsCallingHome}.
While participating in a sensing task, users consume their own resources such as energy and
processing power.
Moreover, users may also have concerns regarding their privacy.
As a result, existing applications and systems may suffer from
insufficient number of participants because it may not be rational for the users to participate.
Thus, 
there is a need to design sensing architectures that can provide appropriate
incentives to the users to motivate their participation. 
Furthermore, users, being selfish in nature, may manipulate 
protocols of the sensing architectures for their own benefits.
Thus, a critical property that any mechanism involving selfish entities
should exhibit is strategy-proofness or truthfulness. 
As has been shown in~\cite{Klemperer00whatreally}, mechanisms 
that are not truthful are prone
to market manipulations and can have inefficient outcomes.

Most past work focusses on sensor management problems without addressing 
selfish concerns of
participants~\cite{Zhao:tsp02,ac10:hoffman,long_analog,engin_localization,ncao_Fusion13}.
Information based measures have been used for sensor management in \cite {Zhao:tsp02,ac10:hoffman}
which maximize the mutual information between the sensor measurements and target state.
Sensor selection strategies that minimize the bound on the estimation error, which is the inverse of Fisher Information \cite{long_analog,engin_localization}
are computationally more efficient than the information based sensor selection
methods \cite{engin_localization}. In \cite{ncao_Fusion13}, the Nondominating
Sorting Genetic Algorithm-II method is employed for the multi-objective optimization based sensor selection problem.
Since selecting a subset of 
informative sensors out of $N$ sensors in the network is an NP-hard combinatorial problem,
in \cite{Joshi:tsp09}, the binary variable sensor selection problem is relaxed and solved using convex optimization.
Transmission of quantized measurements is required in typical WSNs that have limited resources (energy and bandwidth).
This gives rise to the more general problem of bit allocation. Given the total bandwidth constraint, the Fusion Center (FC) determines the optimal bandwidth distribution for the
channels between the sensors and the FC.
In \cite{engin_bit_allocation}, the myopic bandwidth allocation problem is considered and the algorithms to solve the
problem, namely, convex relaxation, A-DP, GBFOS and greedy search, are compared.

Market based 
mechanisms for 
sensor management have started to gain attention only recently \cite{customerSensor,Pricetheorey_Chavali,masazade_icassp13}.
In~\cite{customerSensor}, the authors explored the possibility of using 
economic concepts for sensor management without explicitly formulating a specific problem.
The authors in~\cite{Pricetheorey_Chavali} used the concept of Walrasian equilibrium~\cite{WalrasianEq}
to model market based sensor management.
In~\cite{masazade_icassp13}, the authors also proposed a Walrasian equilibrium based dynamic
bit allocation scheme for target tracking in energy constrained wireless sensor networks (WSNs)
using quantized data. 
However, as shown in~\cite{WalrasianUnstable}, Walrasian markets can be unstable and can fail to converge
to the equilibrium.
Furthermore,
computing the equilibrium prices and allocations can be computationally prohibitive.
Accordingly, the authors (\cite{Pricetheorey_Chavali} and references therein)
propose algorithms to compute an approximate equilibrium.
Moreover, the mechanisms proposed in~\cite{Pricetheorey_Chavali,masazade_icassp13}
are not truthful and are, therefore,
prone to market manipulations.

The main objective of this paper is to design a market-based mechanism \cite{algoMechDesign} to trade information
for tracking a target, with the mechanism being  
computationally efficient, individually-rational (to rationalize user participation),
incentive-compatible (to ensure strategy-proofness), and profitable (to ensure feasibility).
However, as opposed to conventional market scenarios, the problem at hand portrays two unique characteristics--
\begin{inparaenum}[\itshape a\upshape)]
\item Here, the traded commodity in the market is \textit{information}. At what prices would information trade,
    given that the prices users would want to sell their information is dependent on their participatory costs?, and,
\item The information acquisition process is in a resource constrained environment with participants having limited energy,
and bandwidth availability for communication.
How do we allocate resources efficiently in such a resource constrained environment?
\end{inparaenum}
To answer both questions, we propose to use \textit{auctions}~\cite{vijayKrishna,algoMechDesign,optimal_auction_myerson}.
One of the chief virtues of auctions is their ability to determine appropriate prices of
traded commodities~\cite{ascendAuction_peterCrampton}.
Further, there is also substantial agreement among economists that auctions are the best way to allocate resources
in a resource constrained environment~\cite{McMillan1995}.
Essentially,
auctions seek an answer to the basic question `Who should get the resources and at what prices?'


In our prior work \cite{ncao_Globalsip13}, we limited our focus on the design of an incentive-based mechanism for target localization via sensor selection (which is a special case of the bit allocation problem\footnote{It should be noted that, for a given total number of bits per time step that can be transmitted from sensors to the fusion center (FC), dynamic bit allocation
distributes the resources more efficiently, and thus provides better estimation performance as compared to the sensor
selection problem [20].}). In this paper, we focus on the more general problem of designing an incentive-based mechanism for target tracking while considering dynamic bit allocation. 
Specifically, in this paper, we propose a reverse
auction\footnote{A reverse auction is one in which the roles of buyers and seller are reversed.} based mechanism
in which an auctioneer (FC)
conducts an auction
to estimate the target location at each tracking
step by soliciting bids from the selfish users (sensors\footnote{In the rest of the paper,
we refer to users as sensors, unless mentioned explicitly.}). The bids of the sensors reflect how much they value their energy costs. Moreover, the sensors' valuations of their energy costs may also increase as the residual energy depletes, which we also consider in our model. 
Our auction mechanism is comprised of two components-
\begin{inparaenum}[\itshape a\upshape)]
\item \textit{bandwidth allocation function}, which determines how to distribute the limited bandwidth (bits) between the sensors and the FC,
and,
\item \textit{pricing function}, which determines the payment to be made to each user.
\end{inparaenum}
The focus of this paper is to design these two functions.

To address the participatory concerns of the selfish sensors,
we 
design the mechanism so that it is always in the best interest of the selfish sensors to participate in the auction.
Further, our proposed auction mechanism is truthful so that it is not prone to market manipulations.
To implement the proposed auction model in a computationally efficient manner, we use 
dynamic programming by formulating the proposed mechanism as a
multiple-choice knapsack problem (MCKP)~\cite{pisinger2004knapsack,Pisinger1995394}.
As is shown in the paper, the dynamic programming approach finds the exact
equilibrium of our model.
Formally, the key contributions of the paper are as follows.
\begin{itemize}
\item We propose an auction-based market mechanism to trade information
    for tracking a target. The proposed mechanism is 
    computationally efficient, individually-rational, incentive-compatible (truthful),
    and profitable. To the best of our knowledge, we are the first to propose a market mechanism for tracking a target
    using selfish users that exhibits the aforementioned properties.
\item We propose a pseudo-polynomial time procedure to implement the proposed
    auction mechanism using dynamic programming. The dynamic programming approach
    can provably sustain the market at the exact equilibrium. Our solution is thus stable.
\item Via extensive simulations, we show the effectiveness of our proposed mechanism, study its characteristics, and also show the benefits of the ``energy-awareness" of the mechanism when the participatory costs (valuations) of the users are dependent on their residual energy.
\end{itemize}

The rest of the paper is organized as follows. In Section II, we introduce the target tracking background.
In Section III, we introduce the basic assumptions and formulate the problem.
We analyze the incentive-based mechanism in Section IV.
The implementation of our proposed mechanism is discussed in Section V as well as the case where the sensors' valuations are dependent on their residual energy.
Simulation results are presented in Section VI, and we conclude our work in Section VII.

\section{Target Tracking in Wireless Sensor Networks}

\subsection{System Model}
We consider a WSN consisting of $N$ selfish sensors which are uniformly deployed in a square region of interest (ROI) of size $b^2$. Note that, our work can handle any sensor deployment pattern and the uniform sensor deployment is employed here for ease of presentation. We assume that the target and all the sensors are based on flat ground and have the same height, so that we can formulate the problem with a 2-D model. The target is assumed to emit a signal from location $(x_t,y_t)$ at time step $t$, and the FC estimates the position and velocity of the target based on the sensor measurements. The state vector of the target at time $t$ is defined by $\mathbf{x}_t=[x_t\hspace{0.5em} y_t\hspace{0.5em}  \dot{x}_t\hspace{0.5em}  \dot{y}_t]^T$, where ($\dot{x}_t,\dot{y}_t$) are the target velocities. Table \ref{notationTable} gives the notations we use in the paper. The target motion dynamics is defined according to the following linear model,
\begin{equation}
	\label{track_model}
\mathbf{x}_{t+1}=\mathbf{F}\mathbf{x}_t+v_t
\end{equation}where $v_t$ is the zero-mean, Gaussian process noise with covariance matrix $\mathbf{Q}$ where $\mathbf{F}$ and $\mathbf{Q}$ are defined as,
\begin{equation}
	\label{D_Q}
\mathbf{F}= \left[ \begin{array}{cccc}
1 & 0 & \mathcal{D} & 0 \\
0 & 1 & 0 & \mathcal{D} \\
0 & 0 & 1 & 0 \\
0 & 0 & 0 & 1 \end{array} \right],
\mathbf{Q}= \tau \left[ \begin{array}{cccc}
\frac{\mathcal{D}^3}{3} & 0 & \frac{\mathcal{D}^2}{2} & 0 \\
0 & \frac{\mathcal{D}^3}{3} & 0 & \frac{\mathcal{D}^2}{2} \\
\frac{\mathcal{D}^2}{2} & 0 & \mathcal{D} & 0 \\
0 & \frac{\mathcal{D}^2}{2} & 0 & \mathcal{D} \end{array} \right],
\end{equation}

In \eqref{D_Q}, $\mathcal{D}$ and $\tau$ denote the sampling time interval and the process noise parameter respectively. We assume that the FC has complete knowledge of the process model in \eqref{track_model}.

We consider the isotropic power attenuation model for the target as,
\begin{equation}
a_{i,t}^2=\frac{P_0}{1+d_{i,t}^2}
\end{equation}where $a_{i,t}$ is the received signal amplitude at the $i^{th}$ sensor at time step $t$, $P_0$ is the emitted signal power from the target, and $d_{i,t}$ is the distance between the target and the $i^{th}$ sensor at time step $t$, i.e., $d_{i,t} = \sqrt{(x_t-x_i)^2 + (y_t-y_i)^2}$. At time step $t$, the received signal at sensor $i$ is given by
\begin{equation}
	\label{measurement}
z_{i,t}=a_{i,t} + n_{i,t}
\end{equation}
The measurement noise samples $n_{i,t}$ are assumed to be independent across time steps and across sensors and they follow   Gaussian distribution with parameters ${\cal N}(0,\sigma^2)$. In order to reduce the cost of communication, the sensor measurements, $z_{i,t}$'s, are quantized into  $m$-bits before transmission to the FC. The quantized measurement of sensor $i$ at time step $t$, $D_{i,t}$, is defined as:
\begin{equation}
D_{i,t} = \left\{ \begin{array}{rl}
  0 &\mbox{ $-\infty < z_{i,t} <\eta_{1}$} \\
  1 &\mbox{ $\eta_{1} < z_{i,t} < \eta_{2}$}\\
\vdots\\
  L-1 &\mbox{ $\eta_{(L-1)} < z_{i,t} <\infty$ }  \\
       \end{array} \right.
\end{equation}where $\boldsymbol{\eta} = [\eta_{0}, \eta_{1},\ldots,\eta_{L}]^T$ is the set of quantization thresholds with $\eta_{0} = -\infty$ and $\eta_{L} = \infty$ and $L=2^m$ is the number of quantization levels. For simplicity,
the quantization thresholds are designed identically according to the Fisher Information based heuristic quantization as in \cite{tsp06:niu}. Then, given target state $\mathbf{x}_t$ at time step $t$, the probability that $D_{i,t}$ takes value $l$ is,
\begin{equation}
	\label{eq4}
p(D_{i,t}=l|\mathbf{x}_t)=Q\left(\frac{\eta_{l}-a_i}{\sigma}\right)-Q\left(\frac{\eta_{l+1}-a_i}{\sigma}\right)
\end{equation}
where $Q(.)$ is the complementary distribution of the standard normal distribution. Given $\mathbf{x}_t$, the sensor measurements become conditionally independent, so the likelihood function of $\mathbf{D}_t = [D_{1,t}, D_{2,t},...,D_{N,t}]^T$ can be written as,
\begin{equation}
	\label{eq5}
p(\mathbf{D}_t|\mathbf{x}_t) = \prod_{i=1}^{N}p(D_{i,t}|\mathbf{x}_t)
\end{equation}

In our work, we consider that the FC reimburses the sensors for energy spent for transmission. 
By assuming that there are no errors in data transmission, a simple model of energy consumption
of sensor $i$ at time $t$ for transmitting $m$ bits over its distance from the FC $h_{i}$ is
considered as \cite{Heinzelman_energy}
\begin{equation}
\label{energy_model}
E_{i,t}^c (m, h_{i}) = \epsilon_{amp} \times m \times h_{i}^2
\end{equation}
where $\epsilon_{amp}$ is assumed to be $10^{-8} \mathbf{J/bit/m^2}$.
\begin{table}[t]
\center
\begin{tabular}{|c|c|c|} \hline
{\small \bf Notation} & {\small \bf Description}   \\ \hline
{\small $t$}         &  time step  \\ \hline
{\small $N$}         &  no. of sensors  \\ \hline
{\small $b$}         &  size of ROI  \\ \hline
{\small $\mathbf{x}_t$}  &  target state vector at time $t$; to be estimated  \\ \hline
{\small $(x_i, y_i)$}    &  location of sensor $i$ \\ \hline
{\small \multirow{2}{*}{$d_{i,t}$}}  &   distance between the target and\\ &the $i^{th}$ sensor at time step $t$ \\ \hline
{\small $D_{i,t}$}       &  quantized measurement of sensor $i$ at time step $t$ \\ \hline
{\small \multirow{2}{*}{$M$}}       &  total bandwidth constraint of the\\ &system at each time step \\ \hline
{\small $J_{i,t}^D$}           &  FIM obtained from the sensor $i$'s measurement \\ \hline
{\small $J_t^P$}           &  FIM of the \textit{a priori} information  \\ \hline
{\small $\mathbf{v}$}         &  sensors' valuation vector  \\ \hline
{\small $v_{FC}$}         &  FC's valuation per unit information  \\ \hline
{\small \multirow{2}{*}{$T$}}         &  set of all possible combinations\\ &of bidders' value estimates  \\ \hline
{\small \multirow{2}{*}{$T_{-i}$}}         &  set of all possible combinations of\\ &bidders' value estimates except sensor $i$  \\ \hline
{\small $f$}         &  probability distribution of bidders' value estimate \\ \hline
{\small $\mathbf{p}$}         &  payment vector  \\ \hline
{\small $\mathbf{q}$}         &  sensor's bit allocation variable vector  \\ \hline
{\small $E_{i,t}^c$}         &  sensor $i$'s energy consumption at time $t$  \\ \hline
{\small $\mathcal{U}_t^{FC}$}         &  FC's utility at time $t$  \\ \hline
{\small $\mathcal{U}_{i,t}$}         &  sensor $i$'s utility at time $t$  \\ \hline
{\small $\tilde{\mathcal{U}}_{i,t}$}         &  sensor $i$'s utility at time $t$ if it lied  \\ \hline

\end{tabular}
\caption{Notations used in the paper.}
\label{notationTable}
\end{table}

\subsection{Fisher Information with Quantized Measurements}
Posterior Cramer-Rao Lower Bound (PCRLB) provides the lower bound on estimation error variance for a Bayesian estimator \cite{trees2001modulation}, which is represented as the inverse of the Fisher Information Matrix (FIM),
\begin{equation}\label{pcrlb}
E\left\{[\hat{\mathbf{x}}_t - \mathbf{x}_t][\hat{\mathbf{x}}_t - \mathbf{x}_t]^T\right\} \geq J_t^{-1}
\end{equation}
where $\hat{\mathbf{x}}_t$ is the estimate of the target location, and the FIM is a function of the joint probability density function of the sensor measurements and the target location $p(\mathbf{D}_t,\mathbf{x}_t)$,
\begin{align}\label{fim}
J_t &= E[-\Delta_{\mathbf{x}_t}^{\mathbf{x}_t}\log p(\mathbf{D}_t,\mathbf{x}_t)] \\
&= E[-\Delta_{\mathbf{x}_t}^{\mathbf{x}_t} \log p(\mathbf{D}_t|\mathbf{x}_t)] + E[-\Delta_{\mathbf{x}_t}^{\mathbf{x}_t} \log p(\mathbf{x}_t)]\nonumber \\
& = J_t^D + J_t^P \nonumber
\end{align} where, expectation is taken with respect to $p(\mathbf{D}_t,\mathbf{x}_t)$, and $\Delta_{\mathbf{x}}^{\mathbf{x}} \triangleq \bigtriangledown_{\mathbf{x}} \bigtriangledown_{\mathbf{x}}^{T}$ is the second order derivative operator. In \eqref{fim}, the FIM is decomposed into two parts, where, $J_t^D$ represents the FIM corresponding to the sensor measurements, and $J_t^P$ represents the FIM corresponding to the \textit{a priori} information. The FIM corresponding to the sensor measurements, $J_t^D$, can be further written as the summation of each sensor's individual FIM \cite{engin_localization,engin_bit_allocation} as,
\begin{eqnarray}
\label{eq:JD_sum}
&& J_t^D = \sum_{i=1}^N J_{i,t}^D = \sum_{i=1}^N  \int_{\mathbf{x}_t}  J_{i,t}^S (\mathbf{x}_t) p(\mathbf{x}_t) d \mathbf{x}_t
\end{eqnarray} where $J_{i,t}^S(\mathbf{x}_t)$ represents the standard FIM corresponding to sensor $i$ and can be written as,
\begin{align} 
\label{eq:mle_fim}
J_{i,t}^S(\mathbf{x}_t) &= E[-\Delta_{\mathbf{x}_t}^{\mathbf{x}_t} \log p(D_i|\mathbf{x}_t)] = \frac{ 4\kappa_{i,t} a_{i,t}^2}{(1+d_i^2)^2} \times  \\
 &	\left[ \begin{array}{cccc}
	(x_i-x_t)^2  & (x_i-x_t)(y_i-y_t) & 0 & 0\\
	(x_i-x_t)(y_i-y_t)   & (y_i-y_t)^2 & 0 & 0\\
	0 & 0 & 0 & 0\\
	0 & 0 & 0 & 0\\
       \end{array} \right] \nonumber
\end{align}
where $$\kappa_{i,t} = \frac{1}{8\pi \sigma^2} \sum_{l=0}^{L-1}\frac{\left[ e^{-\frac{{(\eta_{l}-a_{i,t})}^2}{2\sigma^2}}-e^{-\frac{{(\eta_{(l+1)}-a_{i,t})}^2}{2\sigma^2}}\right]^2}{p(D_i=l |\mathbf{x}_t)}$$
A detailed derivation of $J_{i,t}^S (\mathbf{x}_t)$ can be found in \cite{engin_localization,engin_bit_allocation}.

\subsection{Particle Filtering based Target Tracking}
In this paper, we employ sequential importance resampling (SIR) particle filtering algorithm \cite{particle} to solve the nonlinear Bayesian filtering problem, where the main idea is to find a discrete representation of the posterior distribution $p(\mathbf{x}_{t}|\mathbf{D}_{1:t})$ by using a set of particles $\mathbf{x}_t^s$ with associated weights $w_t^s$,
\begin{equation}
\label{eq:particle_filter_original}
p(\mathbf{x}_{t}|\mathbf{D}_{1:t}) \approx \sum_{s=1}^{N_s} w_t^s \delta(\mathbf{x}_{t} - \mathbf{x}_{t}^s)
\end{equation} where, $N_s$ denotes the total number of particles. Algorithm \ref{alg1} provides a summary of SIR particle filtering for the target tracking problem, where $T$ denotes the number of time steps over which the target is tracked and $p(\mathbf{D}_{t+1}|\mathbf{x}_{t+1}^{s})$ has been obtained according to (\ref{eq4}) and (\ref{eq5}). Resampling step avoids the situation that all but one of the importance weights are close to zero \cite{particle}.

\begin{algorithm}                      
\caption{SIR Particle Filter for target tracking}          
\label{alg1}                           
\begin{algorithmic}[1]                    
\STATE Set $t = 0$. Generate initial particles $\mathbf{x}_{0}^{s} \sim p(\mathbf{x}_0)$ with $\forall s\;, w_0^s = N_s^{-1}$.
\WHILE {$t \leq T$}
\STATE $\mathbf{x}_{t+1}^{s} = \mathbf{F}\mathbf{x}_{t}^{s} +  \mathbf{\upsilon}_t$ (Propagating particles)
\STATE $p(\mathbf{x}_{t+1}|\mathbf{D}_{1:t}) = \frac{1}{N_s}\sum_{s=1}^{N_s} \delta(\mathbf{x}_{t+1} - \mathbf{x}_{t+1}^{s}) $
\STATE Obtain sensor data $\mathbf{D}_{t+1}$
\STATE $w_{t+1}^{s} \propto  p(\mathbf{D}_{t+1}|\mathbf{x}_{t+1}^{s})$ (Updating weights)
\STATE $w_{t+1}^{s} = \frac{w_{t+1}^{s}}{\sum_{s=1}^{N_s} w_{t+1}^{s}}$ (Normalizing weights)
\STATE $\mathbf{\hat{x}}_{t+1} = \sum_{s=1}^{N_s} w_{t+1}^{s} \mathbf{x}_{t+1}^{s}$
\STATE $\{\mathbf{x}_{t+1}^{s},N_s^{-1}\} = \textrm{Resampling}(\mathbf{x}_{t+1}^{s},w_{t+1}^{s}) $
\STATE $t = t+1$
  \ENDWHILE
\end{algorithmic}
\end{algorithm}


\section{Formulation of the Auction Design Problem}
Our problem belongs to the general area of mechanism design~\cite{algoMechDesign}.
Below we first describe the mechanism design problem in general before formulating our auction in the context of sensor management for tracking.

\subsection{Mechanism Design}
\label{algoMechDesign}
Consider $n$ agents where each agent $i \in \{1,\cdots,n\}$ has some private information which is referred to as his type $t_i$. An output
specification maps each type vector $t = (t_1,\cdots, t_n)$ to a set of allowed outputs.
Depending on his private information, each agent has his own preferences over the possible outputs.
The preferences of agent $i$ are given by a valuation function $v_i$ that assigns a real number $v_i(t_i,q)$ to
each possible output $q$. Each agent $i$ reports his type as $\hat{t}_i$ to the mechanism.
Based on the vector of announced types $\hat{t} = (\hat{t}_1,\cdots,\hat{t}_n)$, the mechanism computes an output
$q(\hat{t})$ and a payment $p_i(\hat{t})$ to each of the agents. The utility of agent
$i$ is $p_i(\hat{t}) + v_i(t_i, q(\hat{t}))$, which the agents wants to optimize.
The following two properties should be exhibited by the mechanism.
\begin{itemize}
\item {\em Incentive Compatibility:} Each agent should be able to maximize his utility by reporting his true type
    $t_i$ to the mechanism so that the mechanism is truthful. In other words,
    \begin{equation}
    p_i(t_{i},\hat{t}_{-i}) + v_i(t_i,q(t_i,\hat{t}_{-i})) \geq p_i(\hat{t}_i,\hat{t}_{-i}) + v_i(t_i,q(\hat{t}_i, \hat{t}_{-i})) \nonumber
    \end{equation}
\item {\em Individual Rationality:} The utility of an agent should be non-negative, so that it is rational for him to participate in the game.
\end{itemize}

\subsection{Our Auction Model}
The sensors, in our model, compete to sell the information contained in their measurements to the FC,
and comprise the set of 
bidders (potential sellers) in the sensor network.
We assume that each bidder $i$ has a valuation $v_i$ per unit of energy,
and that $v_i$ is the true valuation of $i$.
Further, we assume that the FC will derive a benefit from performing the location estimation
and assume that the valuation of the FC per unit of information of the selected sensors is $v_{FC}$~\footnote{$v_{FC}$, for instance, can reflect the
valuation of the entity trying to find the lost/stolen object as discussed
in~\cite{bikeTrack,mobileObjectAlgo,ObjectsCallingHome}}.
The FC is assumed to be unaware of the true valuations of the sensors
so that the sensors have to \textit{advertise} their valuations
at the beginning of the target tracking process
to the FC.
This gives 
the sensors an opportunity to lie about their valuations hoping for an extra benefit.
For instance, a sensor may understate its valuation per unit energy
in the hope of making the FC buy information with finer quantization (larger number of bits), which countervails its loss for announcing valuation lower than the truthful one,
than what the
FC should optimally buy it at.
Or it may exaggerate its valuation
that might increase the payment made to the sensor sufficiently to compensate for any
resulting decrease in the resolution of the information bought.

We 
assume that
the FC's uncertainty about the value estimate of bidder $i$
can be described by a continuous probability distribution $f_i: [a_i,b_i] \rightarrow \mathbf{R}_+$
over a finite interval $[a_i,b_i]$, where $a_i$ is the lowest possible value which $i$ might
assign to its data, and $b_i$ is the highest possible value which $i$ might assign to its data,
and $-\infty \leq a_i \leq b_i \leq \infty$. $F_i: [a_i,b_i] \rightarrow [0,1]$ denotes the
cumulative distribution function, where $F_i(v_i) = \int_{a_i}^{v_i} f_i(t_i)d t_i$. 
We let $T$ denote the set of all possible combinations of bidders' value estimates:
$$T = [a_1, b_1] \times \ldots \times [a_n, b_n]$$
Also, 
for any bidder $i$, the set of all the combinations of the other bidders' value estimates is
$$T_{-i} = [a_1, b_1] \times \ldots \times[a_{i-1}, b_{i-1}] \times [a_{i+1}, b_{i+1}]\times \ldots\times [a_n, b_n]$$
The value estimates of the $N$ sensors are assumed to be statistically independent random variables.
Thus, the joint pdf of the vector $\mathbf{v}=(v_1,\ldots,v_N)$ is
\begin{equation}
f(\mathbf{v}) = \prod_{j \in \left\{1,2,\ldots N\right\}} f_j(v_j)
\end{equation}
We assume that bidder $i$ treats other sensors' value estimates in a similar way as the FC does.
Thus, both the FC and the bidder $i$ consider the joint pdf of the vector of values for
all the sensors other than $i$ $(v_1,\ldots,v_{i-1},v_{i+1},\ldots,v_N)$ to be
\begin{equation}
f_{-i}(\mathbf{v}_{-i}) = \prod_{j \in \left\{1, \ldots, i-1, i+1, N\right\}} f_j(v_j)
\end{equation}

\subsection{Problem Formulation}
\label{pro_form}
Based on the above definitions and assumptions, we consider a direct revelation mechanism, where the bidders simultaneously and confidentially announce their value estimates to the FC. 
The FC then determines the number of bits it should buy from each sensor and how much it should pay them. Thus, our objective is to maximize the FC's utility as a function of the bit allocations and the payment vector. We also assume that the FC and the sensors are risk neutral. By using the trace of the FIM as the metric of tracking performance, the sensors have additively separable utility for money and the commodity (information) being traded \cite{optimal_auction_myerson}.
\subsubsection{Utility Functions}
At time step $t$, we define the expected utility $\mathcal{U}_t^{FC}$ for the FC from the auction mechanism as
\begin{equation}\label{eq1}
\begin{aligned}
\mathcal{U}_t^{FC} (\mathbf{p}, \mathbf{q}) &= \int_{T} \left[v_{FC} \operatorname{tr} \left(\sum\limits_{i=1}^{N} \sum\limits_{m=0}^{M} q_{i,m} (\mathbf{v}) J_{i,t}^D(q_{i,m} = m) \right.\right.\\
&\qquad \qquad \left. \left.+ J_t^P\right) - \sum\limits_{i=1}^{N} p_i(\mathbf{v})\right]f(\mathbf{v})\mathrm{d} \mathbf{v}
\end{aligned}
\end{equation}
where $\mathbf{p} = [p_1,\ldots,p_N]$ is the payment vector and $p_i$ is the expected payment that the FC makes to sensor $i$. $\mathbf{q} = [\mathbf{q}_1^T, \ldots, \mathbf{q}_N^T]^T$ and $\mathbf{q}_i = [q_{i,0}, \ldots,q_{i,m},\ldots,q_{i,M}]^T$ are both Boolean vectors where $\mathbf{q}$ represents the bit allocation state of all the sensors and $\mathbf{q}_i$ represents the bit allocation state of sensor $i$, i.e., $q_{i,m} =1$ when sensor $i$ transmits $m$ bits, and $q_{i,m} = 0$ if sensor $i$ does not transmit its data to the FC in $m$ bits. Thus $\sum\limits_{m=0}^{M} m q_{i,m}$ is the number of bits allocated to sensor $i$. Note that both $\mathbf{p}$ and $\mathbf{q}$ are functions of the vector of announced value estimates $\mathbf{v}=\left[ v_1,\ldots,v_N \right]$, and we ignore the time index $t$ for $\mathbf{p}$, $\mathbf{q}$ and the values $\mathbf{v}$ for notational simplicity.
Since sensor $i$ knows that its value estimate is $v_i$, its expected utility $\mathcal{U}_{i,t} (p_i, \mathbf{q}_i, v_i)$ at time $t$ is described as
\begin{equation}\label{eq2}
\mathcal{U}_{i,t} (p_i,\mathbf{q}_i, v_i) = \int_{T_{-i}} \left[ p_i(\mathbf{v}) -  v_i E_{i,t}^c(\mathbf{q}_i, \mathbf{v}) \right] f_{-i}(\mathbf{v}_{-i}) \mathrm{d} \mathbf{v}_{-i}
\end{equation}
where $\mathrm{d} \mathbf{v}_{-i} = \mathrm{d} v_1 \ldots \mathrm{d} v_{i-1} \mathrm{d} v_{i+1} \ldots \mathrm{d} v_n$. As shown in \eqref{energy_model}, $E_{i,t}^c(\mathbf{q}_i(\mathbf{v}), h_i) = \epsilon_{amp} \times (\sum\limits_{m=0}^{M} m q_{i,m}) \times h_i^2 $, where $h_i$ is not a variable, so here we use simplified notation of $E_{i,t}^c(\mathbf{q}_i(\mathbf{v}), h_i)$ as $E_{i,t}^c(\mathbf{q}_i, \mathbf{v})$.
On the other hand, if sensor $i$ claimed that $w_i$ was its value estimate when $v_i$ was its true value estimate, its expected utility $\tilde{\mathcal{U}}_i$ would be
\begin{equation*}
\tilde{\mathcal{U}}_{i,t} = \int_{T_{-i}} \left[ p_i(w_i, \mathbf{v}_{-i}) -  v_i E_{i,t}^c(\mathbf{q}_i, w_i, \mathbf{v}_{-i})\right] f_{-i}(\mathbf{v}_{-i})\mathrm{d} \mathbf{v}_{-i}
\end{equation*}
where $(w_i, \mathbf{v}_{-i}) = (v_1, \ldots v_{i-1}, w_i, v_{i+1} \ldots v_n)$.

\subsubsection{The Optimization Problem}
Thus, the auction mechanism based bit allocation problem at time step $t$ can be explicitly formulated as follows:
\begin{subequations}\label{eq_opt}
\begin{align}
& \underset{\mathbf{q}}{\text{maximize}} & & \mathcal{U}_t^{FC} (\mathbf{p}, \mathbf{q}) \nonumber\\
& \text{subject to} & & \mathcal{U}_{i,t} (p_i, \mathbf{q}_i, v_i) \geq 0, \hspace{3pt}  i \in \left\{1, \ldots N \right\}\label{eq_opt_a}\\
& & &\qquad \qquad  \qquad \hspace{5pt}  m \in \left\{0, \ldots M \right\}, \hspace{3pt} \forall v_i \in \left[ a_i, b_i \right] \nonumber\\
& & & \mathcal{U}_{i,t} \geq \tilde{\mathcal{U}}_{i,t},  \hspace{10pt}  i \in \left\{1, \ldots N \right\}\label{eq_opt_b}\\
& & & \sum\limits_{i=1}^{N} \sum\limits_{m=0}^{M} m q_{i,m} \leq M\label{eq_opt_c}\\
& & & \sum\limits_{m=0}^{M} q_{i,m} = 1, \hspace{10pt} i \in \left\{1, \ldots N \right\}\label{eq_opt_d} \\
& & & q_{i,m} \in \{0,1\}, \hspace{10pt} i \in \left\{1, \ldots N \right\}, m \in \left\{1, \ldots M \right\}\label{eq_opt_e}
\end{align}
\end{subequations}
Below we describe each constraint in detail.
\begin{itemize}
\item \textit{Individual-Rationality (IR) constraint~\eqref{eq_opt_a}}: We assume that the FC cannot force a sensor to participate in an auction. If it did not participate in the auction, the sensor would not get paid, but also would not have any energy cost, so its utility would be zero. Thus, to guarantee that the sensors will participate in the auction, this condition must be satisfied.
\item \textit{Incentive-Compatibility (IC) constraint~\eqref{eq_opt_b}:} We assume that the FC can not prevent any sensor from lying about its value estimate if the sensor is expected to gain from lying.
Thus, to prevent sensors from lying,
honest responses must form a Nash equilibrium in the auction game.
\item \textit{Bandwidth Limitation (BL) constraint~\eqref{eq_opt_c}:} The FC can buy no more than $M$ bits from all the sensors.
\item \textit{Number of quantization Levels (NQL) constraint~\eqref{eq_opt_d}:}
    Each sensor uses only one quantization level.
\item $q_{i,m}$~\eqref{eq_opt_e} is a Boolean variable.
\end{itemize}


\section{Analysis of the Auction Design Problem}

In this section, we analyze the optimization problem proposed in Section \ref{pro_form}. We define
\begin{equation}\label{eq6}
\mathcal{B}_{i,t} (\mathbf{q}_i,v_i) = \int_{T_{-i}} E_{i,t}^c(\mathbf{q}_i, \mathbf{v}) f_{-i}(\mathbf{v}_{-i}) \mathrm{d} \mathbf{v}_{-i}
\end{equation}
at time step $t$ for any sensor $i$ with value estimate $v_i$. So $\mathcal{B}_{i,t} (q_i,v_i)$ denotes the expected amount of energy that sensor $i$ would spend for communication with the FC conditioned on the valuations of the other sensors $\mathbf{v}_{-i}$.

Our first result is a simplified characterization of the IC constraint of the feasible auction mechanism.

\begin{lem}
The IC constraint holds if and only if the following conditions hold:
\begin{equation}\label{eq7}
1~~if \hspace{15pt} v_i \leq w_i \hspace{15pt} then \hspace{15pt} \mathcal{B}_{i,t} (\mathbf{q}_i,v_i) \geq \mathcal{B}_{i,t} (\mathbf{q}_i,w_i)
\end{equation}
\begin{equation}\label{eq8}
2~~\mathcal{U}_{i,t} (p_i, \mathbf{q}_i, v_i) = \mathcal{U}_{i,t} (p_i, \mathbf{q}_i, b_i) + \int\limits_{v_i}^{b_i} \mathcal{B}_{i,t} (\mathbf{q}_i,v_i) \mathrm{d} v_i
\end{equation}
%
%

\end{lem}
\begin{proof}
We first show the ``only if'' part. Without loss of generality, consider that $v_i \leq w_i$. We first consider the case that bidder $i$ claimed that $w_i$ is his value estimate, when $v_i$ is its true value estimate.
\begin{equation*}
\begin{aligned}
&\mathcal{U}_{i,t} (p_i, \mathbf{q}_i, v_i) = \int_{T_{-i}} \left[ p_i(\mathbf{v}) -  E_{i,t}^c(\mathbf{q}_i, \mathbf{v}) v_i \right] f_{-i}(\mathbf{v}_{-i}) \mathrm{d} \mathbf{v}_{-i} \\
& \geq \int_{T_{-i}} \left[ p_i(\mathbf{v}_{-i}, w_i) -  E_{i,t}^c(\mathbf{q}_i,\mathbf{v}_{-i}, w_i) v_i \right] f_{-i}(\mathbf{v}_{-i}) \mathrm{d} \mathbf{v}_{-i} \\
& = \int_{T_{-i}} \left[ p_i(\mathbf{v}_{-i}, w_i) -  E_{i,t}^c(\mathbf{q}_i,\mathbf{v}_{-i}, w_i) w_i \right] f_{-i}(\mathbf{v}_{-i}) \mathrm{d} \mathbf{v}_{-i} \\
& + \int_{T_{-i}} \left[ w_i E_{i,t}^c(\mathbf{q}_i,\mathbf{v}_{-i}, w_i) \right] f_{-i}(\mathbf{v}_{-i}) \mathrm{d} \mathbf{v}_{-i} \\
& - \int_{T_{-i}} \left[ v_i E_{i,t}^c(\mathbf{q}_i,\mathbf{v}_{-i}, w_i) \right] f_{-i}(\mathbf{v}_{-i}) \mathrm{d} \mathbf{v}_{-i} \\
& = \mathcal{U}_{i,t} (p_i, \mathbf{q}_i, w_i) \\
& \qquad + (w_i - v_i) \int_{T_{-i}} E_{i,t}^c(\mathbf{q}_i,\mathbf{v}_{-i}, w_i) f_{-i}(\mathbf{v}_{-i}) \mathrm{d} \mathbf{v}_{-i}
\end{aligned}
\end{equation*}
So we can get,
\begin{equation}\label{eq10}
 \mathcal{U}_{i,t} (p_i, \mathbf{q}_i, v_i) \geq  \mathcal{U}_{i,t} (p_i, \mathbf{q}_i, w_i) + (w_i - v_i) \mathcal{B}_{i,t} (\mathbf{q}_i,w_i)
\end{equation}

Thus, the IC constraint is equivalent to \eqref{eq10}. We now show that \eqref{eq10} implies \eqref{eq7} and \eqref{eq8}. By switching the roles of $v_i$ and $w_i$, we have
\begin{equation}\label{eq10_2}
 \mathcal{U}_{i,t} (p_i, \mathbf{q}_i, w_i) \geq  \mathcal{U}_{i,t} (p_i, \mathbf{q}_i, v_i) + (v_i - w_i) \mathcal{B}_{i,t} (\mathbf{q}_i,v_i)
\end{equation}
Combining \eqref{eq10} and \eqref{eq10_2}, we can see that
\begin{equation*}
\begin{aligned}
(w_i - v_i) \mathcal{B}_{i,t} (\mathbf{q}_i,w_i)  &\leq \mathcal{U}_{i,t} (p_i, \mathbf{q}_i, v_i) - \mathcal{U}_{i,t} (p_i, \mathbf{q}_i, w_i) \\
&\leq (w_i - v_i) \mathcal{B}_{i,t} (\mathbf{q}_i,v_i)
\end{aligned}
\end{equation*}
from which we can derive \eqref{eq7}.

Define $\delta=w_i-v_i$, these inequalities can be written for any $\delta \to 0$
\begin{equation*}
\begin{aligned}
\delta \mathcal{B}_{i,t} (\mathbf{q}_i,v_i + \delta)  &\leq \mathcal{U}_{i,t} (p_i, \mathbf{q}_i, v_i) - \mathcal{U}_{i,t} (p_i, \mathbf{q}_i, v_i + \delta)\\ 
&\leq \delta \mathcal{B}_{i,t} (\mathbf{q}_i,v_i)
\end{aligned}
\end{equation*}
Thus $\mathcal{B}_{i,t} (\mathbf{x},v_i)$ is a decreasing function and it is, therefore, Riemann integrable. We then write the utility function of sensor $i$ for all $v_i \in \left[ a_i, b_i \right]$ as
\begin{equation*}
\mathcal{U}_{i,t} (p_i, \mathbf{q}_i, v_i) = \mathcal{U}_{i,t} (p_i, \mathbf{q}_i, b_i) + \int\limits_{v_i}^{b_i} \mathcal{B}_{i,t} (\mathbf{q}_i,v_i) \mathrm{d} v_i
\end{equation*}
which gives us \eqref{eq8}.

Now we must show the ``if'' part of Lemma 1, i.e., the conditions in Lemma 1 also imply the IC constraint.
Suppose $v_i \leq w_i$, then \eqref{eq7} and \eqref{eq8} give us:
\begin{equation*}
\begin{aligned}
\mathcal{U}_{i,t} (p_i, \mathbf{q}_i, v_i) &= \mathcal{U}_{i,t} (p_i, \mathbf{q}_i, w_i) + \int\limits_{v_i}^{w_i} \mathcal{B}_{i,t} (\mathbf{q}_i,r_i) \mathrm{d} r_i  \\
&  \geq \mathcal{U}_{i,t} (p_i, \mathbf{q}_i, w_i) + \int\limits_{v_i}^{w_i} \mathcal{B}_{i,t} (\mathbf{q}_i,w_i) \mathrm{d} r_i  \\
&  =  \mathcal{U}_{i,t} (p_i, \mathbf{q}_i, w_i) + (w_i - v_i) \mathcal{B}_{i,t} (\mathbf{q}_i,w_i)
\end{aligned}
\end{equation*}
Similarly, if $v_i \geq w_i$,
\begin{equation*}
\begin{aligned}
\mathcal{U}_{i,t} (p_i, \mathbf{q}_i, v_i) &= \mathcal{U}_{i,t} (p_i, \mathbf{q}_i, w_i) - \int\limits_{w_i}^{v_i} \mathcal{B}_{i,t} (\mathbf{q}_i,r_i) \mathrm{d} r_i  \\
&  \geq \mathcal{U}_{i,t} (p_i, \mathbf{q}_i, w_i) - \int\limits_{w_i}^{v_i} \mathcal{B}_{i,t} (\mathbf{q}_i,w_i) \mathrm{d} r_i  \\
&  =  \mathcal{U}_{i,t} (p_i, \mathbf{q}_i, w_i) + (w_i - v_i) \mathcal{B}_{i,t} (\mathbf{q}_i,w_i)
\end{aligned}
\end{equation*}
So \eqref{eq10} can be derived from \eqref{eq7} and \eqref{eq8}. Thus, the conditions in Lemma 1 also imply the IC constraint. This proves the lemma.
\end{proof}

\subsection{Optimal Auction Based Bit Allocation Mechanism}
Based on Lemma 1, problem \eqref{eq_opt} can be simplified as follows.
\begin{thm}
The optimal auction of \eqref{eq_opt} is equivalent to
\begin{equation}\label{eq11}
\begin{aligned}
& \underset{\mathbf{q}}{\text{maximize}} & & \int_{T} \mathcal{Y}_t(\mathbf{q}, \mathbf{v}) f(\mathbf{v}) \mathrm{d} \mathbf{v}\\
& \text{subject to} & & \sum\limits_{i=1}^{N} \sum\limits_{m=0}^{M} m q_{i,m} \leq M\\
& & & \sum\limits_{m=0}^{M} q_{i,m} = 1, \hspace{10pt} i \in \left\{1, \ldots N \right\} \\
& & & q_{i,m} \in \{0,1\}, \hspace{10pt} i \in \left\{1, \ldots N \right\}, m \in \left\{1, \ldots M \right\}
\end{aligned}
\end{equation}
where $\mathcal{Y}_t(\mathbf{q}, \mathbf{v}) =  v_{FC} \operatorname{tr}\left(\sum\limits_{i=1}^{N} \sum\limits_{m=0}^{M} q_{i,m} (\mathbf{v}) J_{i,t}^D(q_{i,m} = m) + J_t^P\right) - \sum\limits_{i=1}^{N} E_{i,t}^c(\mathbf{q}_i, \mathbf{v}) \left(v_i + \frac{F_i (v_i)}{f_i(v_i)}\right)$ and the payment to sensor $i$ is given by
\begin{equation}\label{eq12}
p_i(\mathbf{v}) =  v_i E_{i,t}^c(\mathbf{q}_i, \mathbf{v}) +  \int\limits_{v_i}^{b_i} E_{i,t}^c(\mathbf{q}_i,\mathbf{v}_{-i}, r_i) \mathrm{d} r_i
\end{equation}
\end{thm}

\begin{proof}
We may write the FC's objective function \eqref{eq1} as
\begin{equation*}
\begin{aligned}
&\mathbf{U}_{FC} (p_i, \mathbf{q}) \\
&=\int_{T} \left[v_{FC} \left(\sum\limits_{i=1}^{N} \sum\limits_{m=0}^{M} q_{i,m} (\mathbf{v}) J_i^D(q_{i,m} = m) + J^P\right)\right.\\
&\qquad \quad  \left.  - \sum\limits_{i=1}^{N}p_i(\mathbf{v})\right]f(\mathbf{v}) \mathrm{d} \mathbf{v} + \sum\limits_{i=1}^{N}\int_{T} v_i E_{i,t}^c(\mathbf{q}_i, \mathbf{v}) f(\mathbf{v}) \mathrm{d} \mathbf{v} \\
& \qquad \quad- \sum\limits_{i=1}^{N}\int_{T} v_i E_{i,t}^c(\mathbf{q}_i, \mathbf{v}) f(\mathbf{v}) \mathrm{d} \mathbf{v} \\
\end{aligned}
\end{equation*}
\begin{equation}\label{eq13}
\begin{aligned}
&=\int_{T} v_{FC} \left(\sum\limits_{i=1}^{N} \sum\limits_{m=0}^{M} q_{i,m} (\mathbf{v}) J_i^D(q_{i,m} = m) + J^P\right) f(\mathbf{v}) \mathrm{d} \mathbf{v} \\
& \qquad \quad- \sum\limits_{i=1}^{N}\int_{T} v_i E_{i,t}^c(\mathbf{q}_i, \mathbf{v}) f(\mathbf{v}) \mathrm{d} \mathbf{v}  \\
&\qquad \quad- \left[ \sum\limits_{i=1}^{N}\int_{T} \left( p_i(\mathbf{v}) - v_i E_{i,t}^c(\mathbf{q}_i, \mathbf{v})  \right) f(\mathbf{v}) \mathrm{d} \mathbf{v} \right]
\end{aligned}
\end{equation}

By \eqref{eq7} of Lemma 1, we know that:
\begin{equation*}
\begin{aligned}
&\int_{T} \left( p_i(\mathbf{v}) - v_i E_{i,t}^c(\mathbf{q}_i, \mathbf{v})  \right) f(\mathbf{v}) \mathrm{d} \mathbf{v} \\
& = \int\limits_{a_i}^{b_i} \mathcal{U}_{i,t} (p_i, \mathbf{q}_i, v_i) f_i(v_i) \mathrm{d} v_i \\
& = \int\limits_{a_i}^{b_i} \left(\mathcal{U}_{i,t} (p_i, \mathbf{q}_i, b_i) + \int\limits_{v_i}^{b_i} \mathcal{B}_{i,t} (\mathbf{q}_i,w_i) \mathrm{d} w_i \right) f_i(v_i) \mathrm{d} v_i \\
\end{aligned}
\end{equation*}
\begin{equation}\label{eq14}
\begin{aligned}
&= \mathcal{U}_{i,t} (p_i, \mathbf{q}_i, b_i) + \int\limits_{a_i}^{b_i}  \int\limits_{a_i}^{w_i}f_i(v_i) \mathcal{B}_{i,t} (\mathbf{q}_i,w_i) \mathrm{d} v_i \mathrm{d} w_i   \\
& = \mathcal{U}_{i,t} (p_i, \mathbf{q}_i, b_i) + \int\limits_{a_i}^{b_i} F_i (w_i) \mathcal{B}_{i,t} (\mathbf{q}_i,w_i) \mathrm{d} w_i  \\
& = \mathcal{U}_{i,t} (p_i, \mathbf{q}_i, b_i) + \int_{T} F_i (v_i) E_{i,t}^c(\mathbf{q}_i, \mathbf{v})  f_{-i}(\mathbf{v}_{-i}) \mathrm{d} \mathbf{v}
\end{aligned}
\end{equation}

Substituting \eqref{eq14} into \eqref{eq13} gives us:
\begin{equation*}
\begin{aligned}
&\mathcal{U}_t^{FC} (\mathbf{p}, \mathbf{q}) \\
&= \int_{T} \left[ v_{FC} \left(\sum\limits_{i=1}^{N} \sum\limits_{m=0}^{M} q_{i,m} (\mathbf{v}) J_i^D(q_{i,m} = m) + J^P\right) \right.\\
&\qquad \quad \left.- \sum\limits_{i=1}^{N} v_i E_{i,t}^c(\mathbf{q}_i, \mathbf{v}) \right] f(\mathbf{v}) \mathrm{d} \mathbf{v} - \sum\limits_{i=1}^{N} \mathcal{U}_{i,t} (p_i, \mathbf{q}_i, b_i)\\
& \qquad \quad - \sum\limits_{i=1}^{N} \int_{T} F_i (v_i) E_{i,t}^c(\mathbf{q}_i, \mathbf{v})  f_{-i}(\mathbf{v}_{-i}) \mathrm{d} \mathbf{v}   \\
\end{aligned}
\end{equation*}
\begin{equation}\label{eq15}
\begin{aligned}
&~=  \int_{T} \left[ v_{FC} \left(\sum\limits_{i=1}^{N} \sum\limits_{m=0}^{M} q_{i,m} (\mathbf{v}) J_i^D(q_{i,m} = m) + J^P\right) \right.\\
&\qquad \left.- \sum\limits_{i=1}^{N} v_i E_{i,t}^c(\mathbf{q}_i, \mathbf{v}) \right] f(\mathbf{v}) \mathrm{d} \mathbf{v} \\
&\qquad- \sum\limits_{i=1}^{N} \int_{T} F_i (v_i) E_{i,t}^c(\mathbf{q}_i, \mathbf{v}) \frac{f(\mathbf{v})}{f_i(v_i)}  \mathrm{d} \mathbf{v}  - \sum\limits_{i=1}^{N} \mathcal{U}_{i,t} (p_i, \mathbf{q}_i, b_i)
\end{aligned}
\end{equation}
In \eqref{eq15}, $\mathbf{p}$ appears only in the last term of the objective function. Also, by the IR constraint, we know that 
\begin{equation*}
\mathcal{U}_{i,t} (p_i, \mathbf{q}_i, b_i) \geq 0, \hspace{15pt} i \in \left\{1, \ldots N \right\}
\end{equation*}
Thus, to maximize \eqref{eq15} subject to the constraints, we must have
\begin{equation*}
\mathcal{U}_{i,t} (p_i, \mathbf{q}_i, b_i) = 0, \hspace{15pt} i \in \left\{1, \ldots N \right\}
\end{equation*}
Combining this condition with \eqref{eq2}, \eqref{eq6} and \eqref{eq8}, we get
\begin{equation*}
\begin{aligned}
\mathcal{U}_{i,t} (p_i, \mathbf{q}_i, v_i) &= \int\limits_{v_i}^{b_i} \mathcal{B}_{i,t} (\mathbf{q}_i,v_i) \mathrm{d} v_i\\
&= \int\limits_{v_i}^{b_i} \int_{T_{-i}} E_{i,t}^c(\mathbf{q}_i, \mathbf{v}) f_{-i}(\mathbf{v}_{-i}) \mathrm{d} \mathbf{v}_{-i} \mathrm{d} v_i \\
&= \int_{T_{-i}} \int\limits_{v_i}^{b_i} E_{i,t}^c(\mathbf{q}_i, \mathbf{v}_{-i}, r_i) \mathrm{d} r_i f_{-i}(\mathbf{v}_{-i}) \mathrm{d} \mathbf{v}_{-i} \\
&= \int_{T_{-i}} \left[ p_i(\mathbf{v}) -  v_i E_{i,t}^c(\mathbf{q}_i, \mathbf{v}) \right] f_{-i}(\mathbf{v}_{-i}) \mathrm{d} \mathbf{v}_{-i}
\end{aligned}
\end{equation*}
where the last two equations give the formulation of the payment in \eqref{eq12}. Thus, if the FC pays each sensor according to Equation \eqref{eq12}, then the IR constraint
is satisfied, as well as the best possible value of the last term in \eqref{eq15} is obtained, which is zero. So we can simplify the objective function of our optimization problem to \eqref{eq14} subject to the three remaining constraints. Thus, Theorem 4.1 follows. 
\end{proof}

\section{Implementation of the Proposed Mechanism}
In this section, we consider the algorithm to obtain the solution for the proposed mechanism. We first study the optimal algorithm to solve our optimization problem in \eqref{eq16}, and then the case when sensors' valuations are dependent on their residual energy.

\subsection{Multiple-Choice Knapsack Problems}
The knapsack problem is one of the most important problems in discrete programming \cite{dudzinski1987exact},
and it has been intensively studied for both its theoretical importance and its applications in industry
and financial management. The knapsack problem can be described as: given a set of $n$ items with profit $p_i$ and weight $w_i$ and a knapsack with capacity $c$, select a subset of the items so as to maximize the total profit of the knapsack while the total weights does not exceed $c$
\begin{equation}
\begin{aligned}
& \underset{x_{i}}{\text{maximize}} & & \sum_{i=1}^{n} p_{i}x_{i} \\
& \text{subject to} & & \sum_{i=1}^{n} w_{i} x_{i} \leq c\\
& & & x_{i} \in \{0,1\}, \hspace{10pt} i \in \left\{1, \ldots N \right\},
\end{aligned}
\end{equation}
There are several types of problems in the family of the knapsack problems.
The multiple-choice knapsack problem (MCKP) occurs when the set of items is partitioned into classes and the binary choice of taking an item is replaced by the selection of exactly one item out of each class of items \cite{pisinger2004knapsack}.
Assume that $m$ classes $N_1,\ldots,N_m$ of items are to be packed in a knapsack with capacity $c$. Each item $j \in N_i$ has a profit $p_{i,j}$ and weight $w_{i,j}$. The problem is how to choose one item from each class to maximize the total profit of the knapsack while the total weight does not exceed $c$. The binary variables $x_{i,j}$ are introduced to represent that item $j$ is taken from class $N_i$, the MCKP is formulated as \cite{pisinger2004knapsack} \cite{Pisinger1995394}:
\begin{equation}
\begin{aligned}
& \underset{x_{i,j}}{\text{maximize}} & & \sum_{i=1}^{m} \sum_{j\in N_i} p_{i,j}x_{i,j} \\
& \text{subject to} & & \sum_{i=1}^{m} \sum_{j\in N_i} w_{i,j} x_{i,j} \leq c\\
& & & \sum_{j\in N_i} x_{i,j} = 1, \hspace{10pt} i \in \left\{1, \ldots m \right\} \\
& & & x_{i,j} \in \{0,1\}, \hspace{10pt} i \in \left\{1, \ldots N \right\}, m \in N_i
\end{aligned}
\end{equation}
where $p_{i,j}$, $w_{i,j}$ and $c$ are assumed to be nonnegative integers, with class $N_i$ having size $n_i$ so that the total number of items is $n = \sum_{i=1}^{m} n_i$. By formulating a recursion form, the MCKP can be solved optimally by the dynamic programming method in pseudo-polynomial time with acceptable computation cost when the
number of sensors and the bit constraint are not large.

\subsection{Optimal Solution by Dynamic Programming}
\label{Solve_DP}
 Substituting \eqref{energy_model} into \eqref{eq11}, the objective function $\mathcal{Y}_t$ becomes:
\begin{equation}
\begin{aligned}
& \mathcal{Y}_t(\mathbf{q}, \mathbf{v}) \\
& = v_{FC} \left(\sum\limits_{i=1}^{N} \sum\limits_{m=0}^{M} q_{i,m} (\mathbf{v}) \operatorname{tr} \left( J_{i,t}^D(q_{i,m} = m)\right) + \operatorname{tr} \left(J_t^P\right)\right) \\
&\qquad  - \sum\limits_{i=1}^{N} \left(v_i + \frac{F_i (v_i)}{f_i(v_i)}\right) \left(\left(\sum\limits_{m=0}^{M} m q_{i,m}\right)\epsilon_{amp}h_i^2\right)\\
& = \sum\limits_{i=1}^{N} \sum\limits_{m=0}^{M} q_{i,m} (\mathbf{v}) \left[ v_{FC} \operatorname{tr} \left(J_{i,t}^D(q_{i,m} = m)\right) \right.\\
&\qquad \left. - m \epsilon_{amp}h_i^2 \left(v_i + \frac{F_i (v_i)}{f_i(v_i)}\right)\right] + v_{FC}  \operatorname{tr} \left(J_t^P\right)
\end{aligned}
\end{equation}
where the last term is not subject to the solutions of the optimization problem. Thus, by denoting $V_{i,m} = v_{FC} \operatorname{tr} \left(J_{i,t}^D(q_{i,m} = m)\right) - m \epsilon_{amp}h_i^2 \left(v_i + \frac{F_i (v_i)}{f_i(v_i)}\right)$, the optimization problem in \eqref{eq11} can be written as:
\begin{equation}\label{eq16}
\begin{aligned}
& \underset{\mathbf{q}}{\text{maximize}} & & \int_{T} \left[\sum\limits_{i=1}^{N} \sum\limits_{m=0}^{M} V_{i,m} q_{i,m} \right] f(\mathbf{v}) \mathrm{d} \mathbf{v}\\
& \text{subject to} & & \sum\limits_{i=1}^{N} \sum\limits_{m=0}^{M} m q_{i,m} \leq M\\
& & & \sum\limits_{m=0}^{M} q_{i,m} = 1, \hspace{10pt} i \in \left\{1, \ldots N \right\} \\
& & & q_{i,m} \in \{0,1\}, \hspace{10pt} i \in \left\{1, \ldots N \right\}, m \in \left\{1, \ldots M \right\}
\end{aligned}
\end{equation}
Observe that, given $\mathbf{v}$, \eqref{eq16} is a Multiple Choice Knapsack Problem (MCKP), which is an extension of the Knapsack Problem (KP) \cite{pisinger2004knapsack}. We interpret our optimal auction based bit allocation problem as a MCKP as follows: In the WSN consisting of $N$ sensors, information to be transmitted by each sensor $i$ has $M+1$ variants (bits) where the $m$-th variant has weight $w_{i,m} = m$ and utility value $V_{i,m}$. As the network can carry only a limited capacity $M$, the objective is to select one variant of each sensor such that the overall utility value is maximized without exceeding the capacity constraint.

The MCKP can be solved by the dynamic programming approach in pseudo polynomial time with $O(NM)$ operations \cite{dudzinski1987exact}, \cite{pisinger2004knapsack}. Let $b_l(y)$ denote the optimal solution to the MCKP defined on the first $l$ sensors with restricted capacity $y$
\begin{equation}
\begin{aligned}
&b_l(y) = \\
&\text{max} \left\{\sum\limits_{i=1}^{l} \sum\limits_{m=0}^{M} q_{i,m} V_{i,m} \middle |
													\begin{array}{l}
													\sum\limits_{i=1}^{l} \sum\limits_{m=0}^{M} m q_{i,m} \leq y, \\
													\sum\limits_{m=0}^{M} q_{i,m} = 1, \hspace{3pt} i \in \left\{1, \ldots l \right\},\\
													q_{i,m} \in \{0,1\}, \hspace{3pt} i \in \left\{1, \ldots l \right\}, \\
													\qquad \quad \qquad m \in \left\{0, \ldots M \right\}
													\end{array}\right\}
													\end{aligned}
\end{equation}
and we assume that $b_l(y) = -\infty$ if $y \leq 0$, $l >0$ or $y <0$, $l=0$. Initially we set $b_0(y) = 0$ for $y = 0, \ldots M$. We use the following recursion to compute $b_l(y)$ for $l = 1, \ldots, N$:
\begin{equation}
b_l(y) = \underset{k = 0,\ldots,y}{\text{max}} \left\{ b_{l-1} (y-k) + V_{l,k} \right\}
\end{equation}
$$\vdots$$
$$b_N(M) = \underset{k = 0,\ldots,M}{\text{max}} \left\{ b_{N-1} (M-k) + V_{N,k} \right\}$$

\begin{figure}
\centering
\includegraphics[width=0.75\columnwidth]{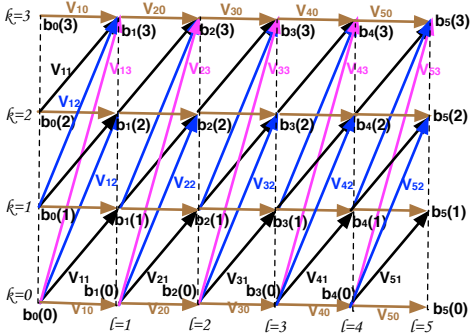}
\caption{Trellis of the dynamic programming algorithm for time step $t$}
\label{Fig_DP}
\end{figure}

To explain the dynamic programming algorithm, we construct the trellis for $N+1$ stages and $M+1$ states associated with each stage \cite{engin_bit_allocation}. Fig. \ref{Fig_DP} gives an example trellis for $N=5$ and $M=3$, which contains 6 stages and 4 states. For example, $b_1(1) = \text{max} \left\{ b_{0} (1) + V_{1,0}, b_0(0) + V_{1,1} \right\}$ and $b_3(2) = \text{max} \left\{ b_2 (2) + V_{3,0}, b_2(1) + V_{3,1}, b_2(0) + V_{3,2} \right\}$. Thus, the optimal solution is found as $b = b_N(M)$. Note that to get the optimal bit allocation, the solution $\mathbf{q}$ needs to be recorded at each step corresponding to the optimal $b_l(y)$. On the other hand, the dynamic programming algorithm for our optimization problem is pseudo polynomial and has the complexity $O(NM)$ \cite{Pisinger1995394}. Therefore, the optimality of the problem \eqref{eq16} is guaranteed and the rationality and the truthfulness properties of our incentive-based mechanism are maintained.

The payment to each sensor can be calculated from \eqref{eq12}. The key point is to find the thresholds between $v_i$ and $b_i$ above which the sensors will be assigned different number of bits compared to the original optimal solution of \eqref{eq11}. The pseudo-code of the detailed algorithm is presented in Algorithm \ref{alg:payment}.
\begin{algorithm}[t]
\caption{Payment Calculation}
\label{alg:payment}
\begin{algorithmic}[1]
\FOR {$i = 1:N$}
	\STATE Calculate the total number of bits sensor $i$ is assigned to through the optimal solution $\mathbf{q}^0_i$ of \eqref{eq11}, $R_i = \sum\limits_{m=0}^{M} m q_{i,m}$.
	\STATE Set the value estimate of sensor $i$ to $b_i$, run our mechanism again
		\IF {sensor $i$ is still assigned $R_i$ bits}
				\STATE $p_i = b_i E_{i,t}^c(\mathbf{q}^0_i, \mathbf{v})$
		\ELSE
						
				\STATE There must be at least one point between $v_i$ and $b_i$ above which less than $R_i$ will be assigned to sensor $i$. We apply bisection method to find the thresholds.
%
%
				\REPEAT \STATE{Bisection method.}
				\UNTIL {finding the point above which sensor $i$ is not assigned any number of bits.}\\
				\STATE Assume there are $n$ thresholds totally. Note that $n \leq R_i$.
				\STATE Denote the thresholds as $w^1_i, \ldots w^{n}_i$ and the corresponding solution vectors as $\mathbf{q}_i^1, \ldots, \mathbf{q}_i^{n-1}$.\\
				\STATE Then $p_i = w^1_i E_{i,t}^c(\mathbf{q}^0_i, \mathbf{v}) + (w^2_i - w^1_i)E_{i,t}^c(\mathbf{q}_i^{1}, \mathbf{v}) + \ldots + (w^n_i - w^{n-1}_i)E_{i,t}^c(\mathbf{q}_i^{n-1}, \mathbf{v})$.
		\ENDIF
\ENDFOR
\end{algorithmic}
\end{algorithm}




\subsection{Residual Energy Dependent Valuations}
\label{Energy_Efficiency}
So far, we have assumed that the valuation of the sensors are invariant
of the amount of residual energy of the sensors over time. We now relax this assumption and consider that the (true) valuations of the sensors are dependent on their residual energy.
Therefore, the remaining energy of the sensors are included in their utility functions, 
\begin{equation}\label{eq2_ee}
\begin{aligned}
\hat{\mathcal{U}}_{i,t} (p_i,\mathbf{q}_i, v_i) &= \int_{T_{-i}} \left[ p_i(\mathbf{v}) -  v_i g(e_{i,t-1}) E_{i,t}^c(\mathbf{q}_i, \mathbf{v}) \right] \\
&\qquad \quad \qquad \quad \qquad \quad \quad f_{-i}(\mathbf{v}_{-i}) \mathrm{d} \mathbf{v}_{-i}
\end{aligned}
\end{equation}
where $e_{i,t-1}$ is the remaining energy of sensor $i$ at the beginning of time
$t-1$, i.e., $e_{i,t-1} = e_{i,t-2} - E_{i,t-1}^c$,
so that including $g(e_{i,t-1})$ makes the problem more general, and
the new objective function $\hat{\mathcal{Y}}_t$ of \eqref{eq11} becomes
\begin{equation}
\begin{aligned}
\hat{\mathcal{Y}}_t &= v_{FC} \operatorname{tr}\left(\sum\limits_{i=1}^{N} \sum\limits_{m=0}^{M} q_{i,m} (\mathbf{v}) J_{i,t}^D(q_{i,m} = m) + J_t^P\right) \\
&\qquad \quad - \sum\limits_{i=1}^{N} g(e_{i,t-1}) E_{i,t}^c(\mathbf{q}_i, \mathbf{v}) \left(v_i + \frac{F_i (v_i)}{f_i(v_i)}\right)
\end{aligned}
\end{equation}
and the corresponding value of $V_{i,m}$ in \eqref{eq16} is 
\begin{equation*}
\begin{aligned}
\hat{V}_{i,m} &= v_{FC} \operatorname{tr} \left(J_{i,t}^D(q_{i,m} = m)\right) \\
&\qquad \quad- h(e_{i,t-1}) m \epsilon_{amp}h_i^2 \left(v_i + \frac{F_i (v_i)}{f_i(v_i)}\right)
\end{aligned}
\end{equation*}
Referring to Section \ref{Solve_DP}, we can find that the target tracking problem with residual energy based valuation can also be mapped to a MCKP and solved by the dynamic programming method in pseudo-polynomial time.
We assume that the FC knows the energy status of all the sensors at each time step, so the FC and the sensors decide how the value estimate of the sensors change with their remaining energy at the beginning of the tracking task.

\section{Simulation Experiments}
In this section, we study the dynamics of our proposed incentive-based target tracking mechanism in a sensor network. In the experiments, $N = 25$ sensors are deployed uniformly in the ROI with the size $50m \times 50m$ and the FC is located at $x_{FC} = -22, y_{FC} = 20$. Note that our model can handle any sensor deployment pattern as long as the sensor locations are known to the FC in advance. The signal power at distance zero is $P_0 = 1000$. The target motion follows the white noise acceleration model with $\tau = 2.5 \times 10^{-3}$. The variance of the measurement noise is selected as $\sigma = 1$. The prior distribution about the state of the target, $p(\mathbf{x}_0)$, is assumed to be Gaussian with the covariance matrix $\Sigma_0 = diag [\sigma_{\mathbf{x}}^2~\sigma_{\mathbf{x}}^2~0.01~0.01]$ where $3\sigma_{\mathbf{x}} = 2$ so that the initial location of the target stays in the ROI with high probability. The pdf of the value estimate of sensor $i$, $v_i$, is assumed to be uniformly distributed between $a_i$ and $b_i$ with $a_i = 0.1$ and $b_i = 1$, and the value estimate of the FC is assumed to be $v_{FC} = 1$. The performance of the target location estimator is determined in terms of the mean square error (MSE) at each time step over $100$ Monte Carlo trials and the number of particles of each Monte Carlo trial is $N_s = 5000$.
\begin{figure*}[htb]
\begin{center}
\subfigure[]{
\includegraphics[%
  width=0.4\textwidth, height = 0.3\textwidth]{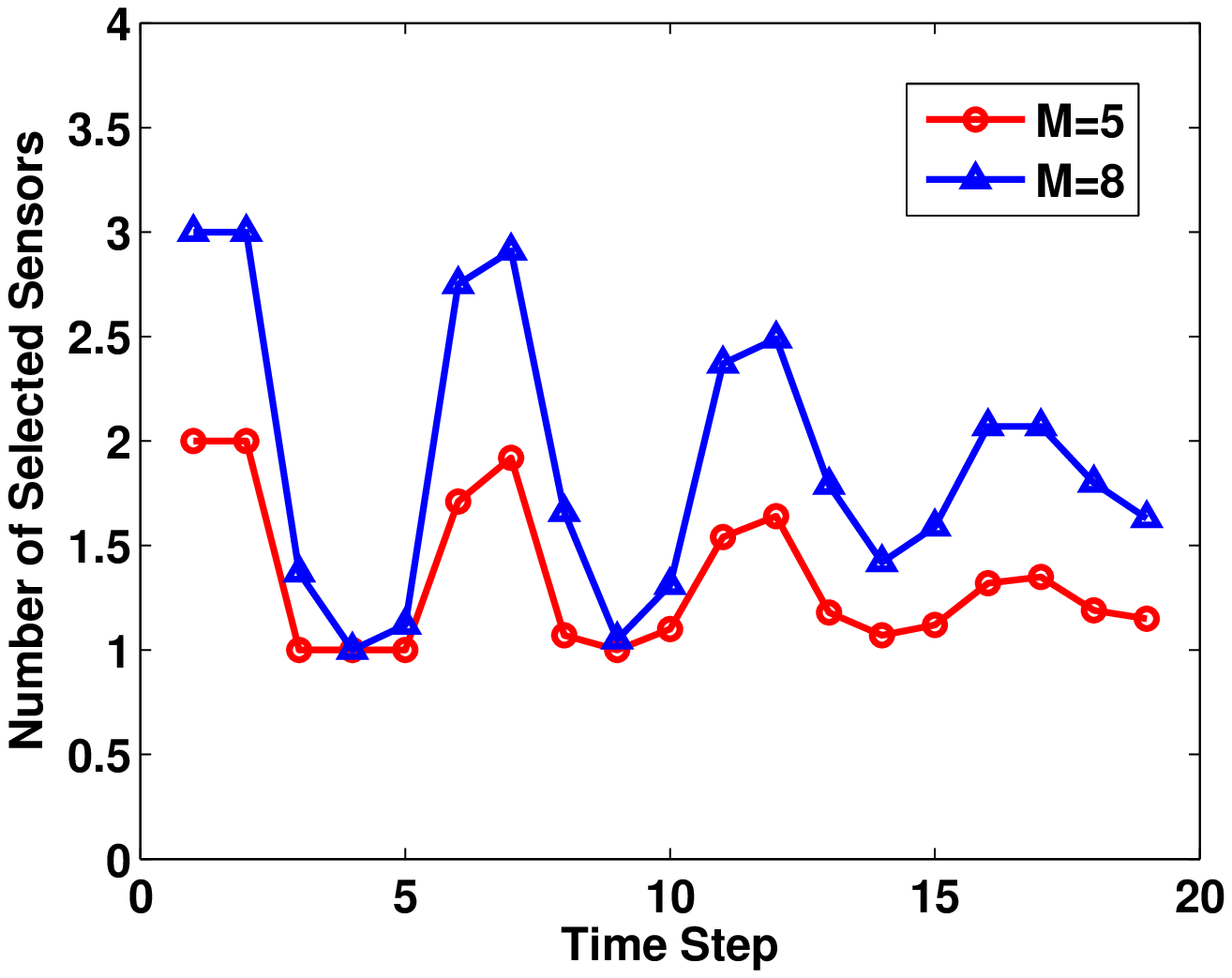}
\label{fig_sensor} }
\subfigure[]{
\includegraphics[%
  width=0.4\textwidth, height = 0.3\textwidth]{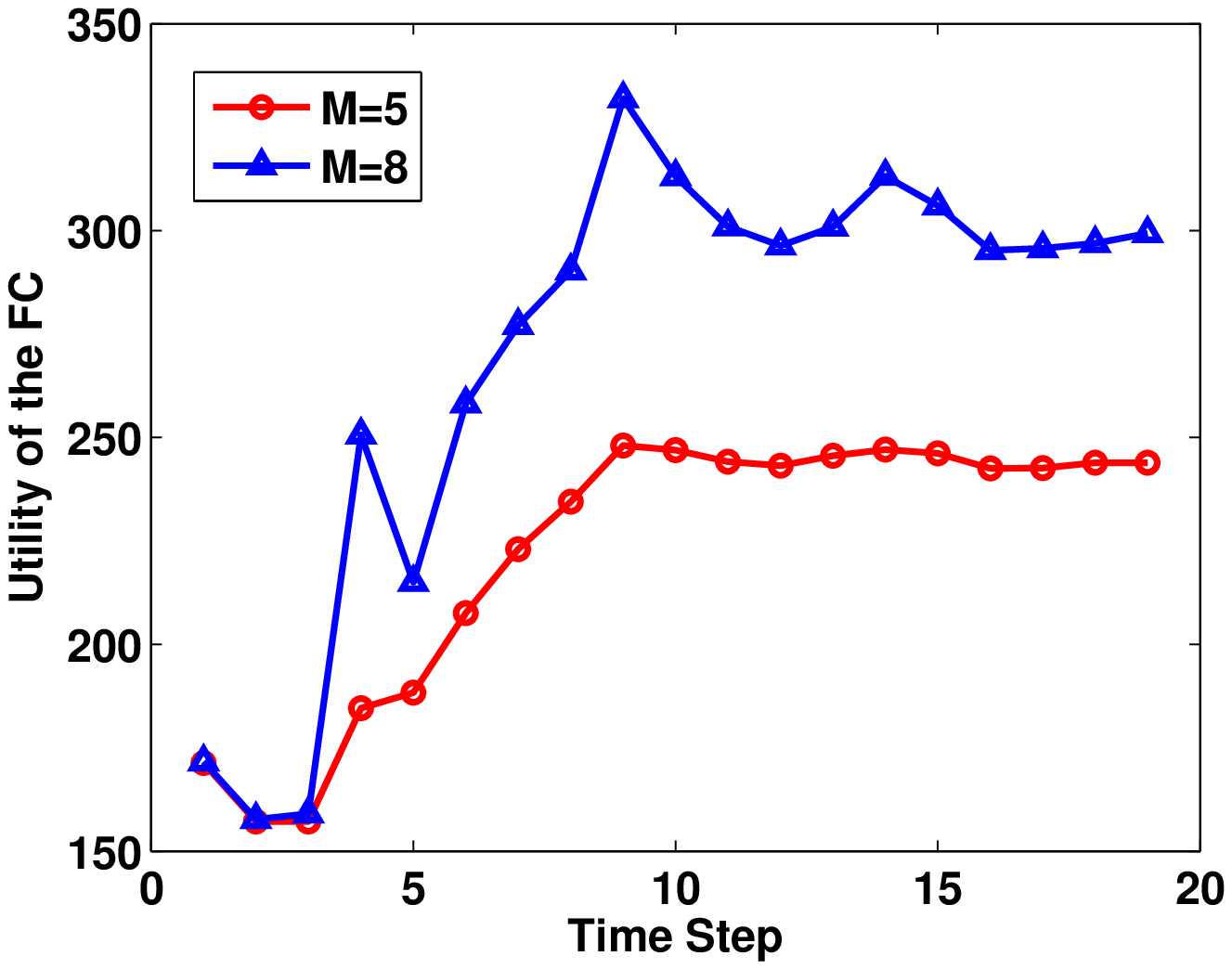}
\label{fig_u} }\\
\subfigure[]{
\includegraphics[%
  width=0.4\textwidth, height = 0.3\textwidth]{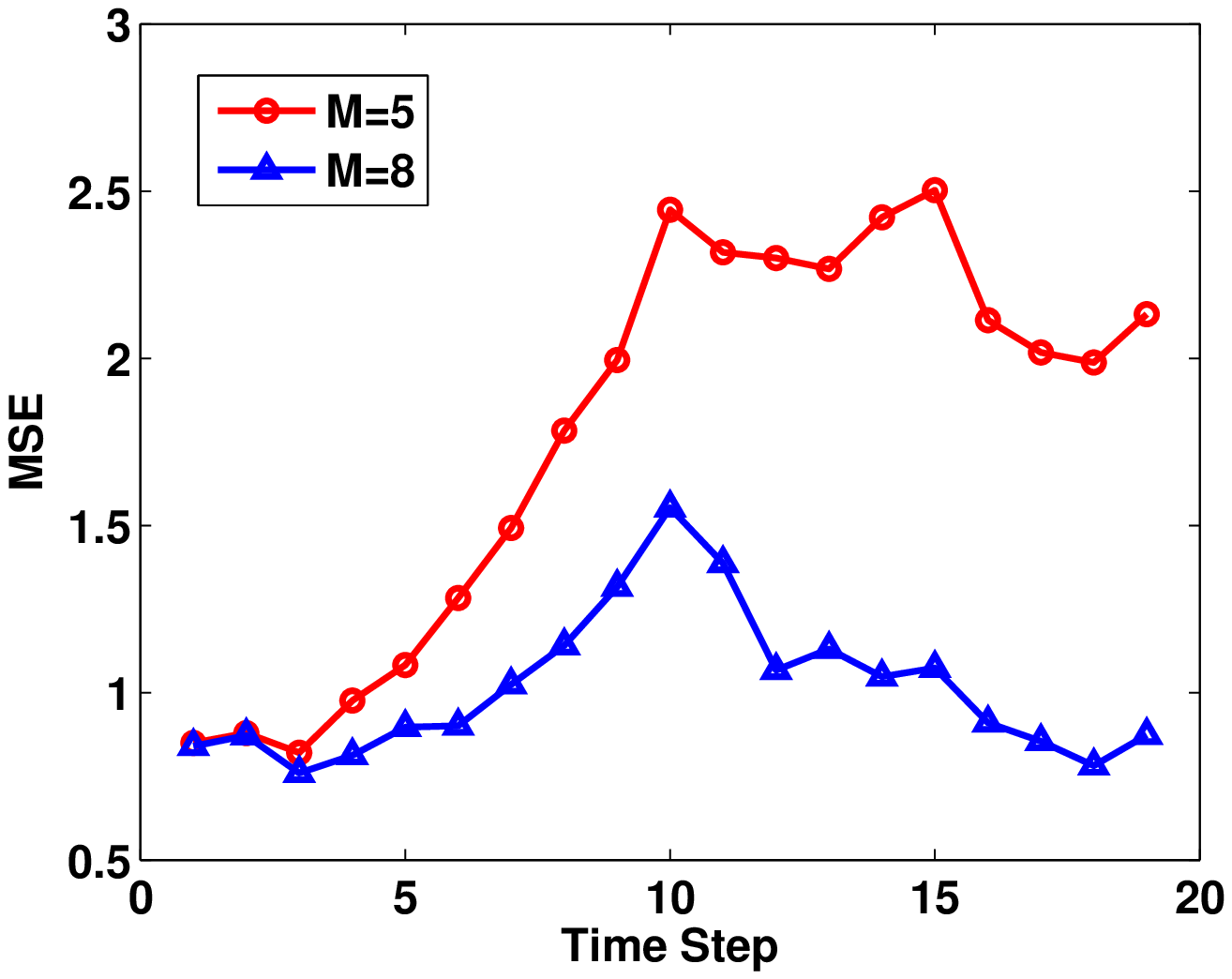}
\label{fig_mse} }
\subfigure[]{
\includegraphics[%
  width=0.4\textwidth, height = 0.3\textwidth]{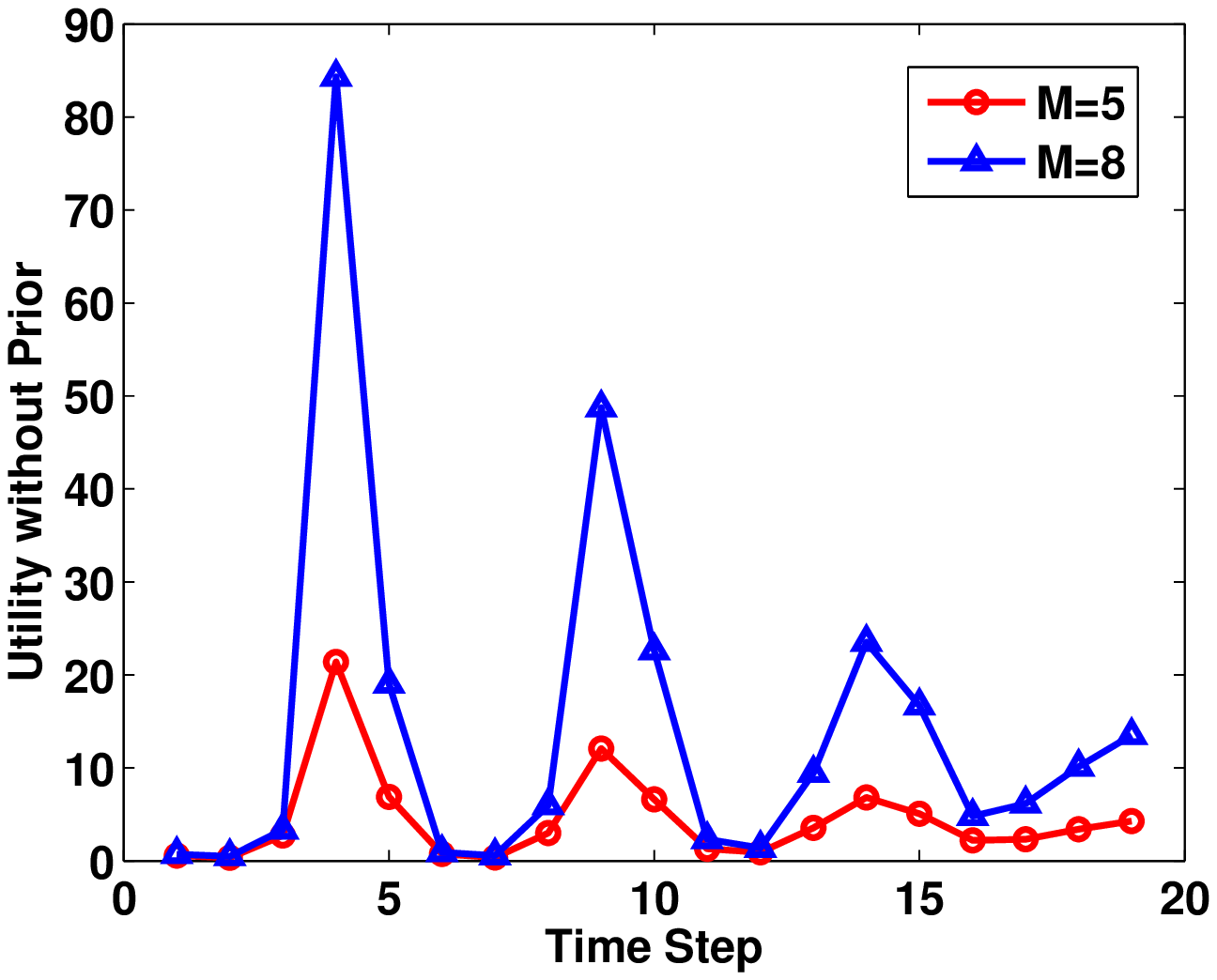}
\label{fig_u2} }
\caption{Bit allocation with $M=5$ and $M=8$ (a) The number of selected sensors. (b) The utility of the FC. (c) MSE at each time step. (d) The utility of the FC with prior information excluded.}
\label{M5_M8}
\end{center}
\end{figure*}

We first consider the implementation of our optimal auction based target tracking procedure as illustrated in Section V-B, where the valuations of
the sensors do not vary with their residual energy.
In the target motion model, measurements are assumed to be taken at regular intervals of $\mathcal{D} = 1.25$ seconds and we observe the target for 20 time steps.
The mean of the prior distribution about the target state is assumed to be $\mu_0 = [-23~-23~2~2]^T$.
Two different values of $M$, namely 5 and 8, are considered to examine the impact of total number of available bits. In Fig. \ref{fig_sensor}, we show the number of sensors that are selected at each time step. And the corresponding tracking MSE is shown in Fig. \ref{fig_mse}. We can see that around time steps 4, 9, 14 and 19, the target is relatively close to some sensor,
and fewer sensors are activated. When the target is not relatively close to any
sensor in the network, during time periods 5-8, 12-13 and 18-19, the uncertainty about the target is relatively high, which increases the estimation error, so that more sensors are activated. 
Fig. \ref{fig_u} shows the total utility of the FC at each time step. 
Note that because of the accumulated information, the utility of the FC increases as time goes by and saturates during the last ten time steps. In Fig. \ref{fig_u2}, we also show the utility of the FC when the term due to prior FIM, $J_t^P$, is not included in the expression for the utility function given in \eqref{eq1}. Due to the low noise environment and the accumulation of the information, $J_t^P$ contains more information and the contribution of the data to the utility function as a function of time diminishes. This is evident in Fig. \ref{fig_u2} in that we observe a decreasing trend of the utility function as a function of time.
Moreover, for all the results, we observe that when we have more number of bits (resources) to allocate, the performance in terms of tracking performance and the gains of the FC is better, i.e., results for $M=8$ are better than those for $M=5$.  


\begin{figure*}[htb]
\begin{center}
\subfigure[]{
\includegraphics[%
  width=0.4\textwidth, height = 0.3\textwidth]{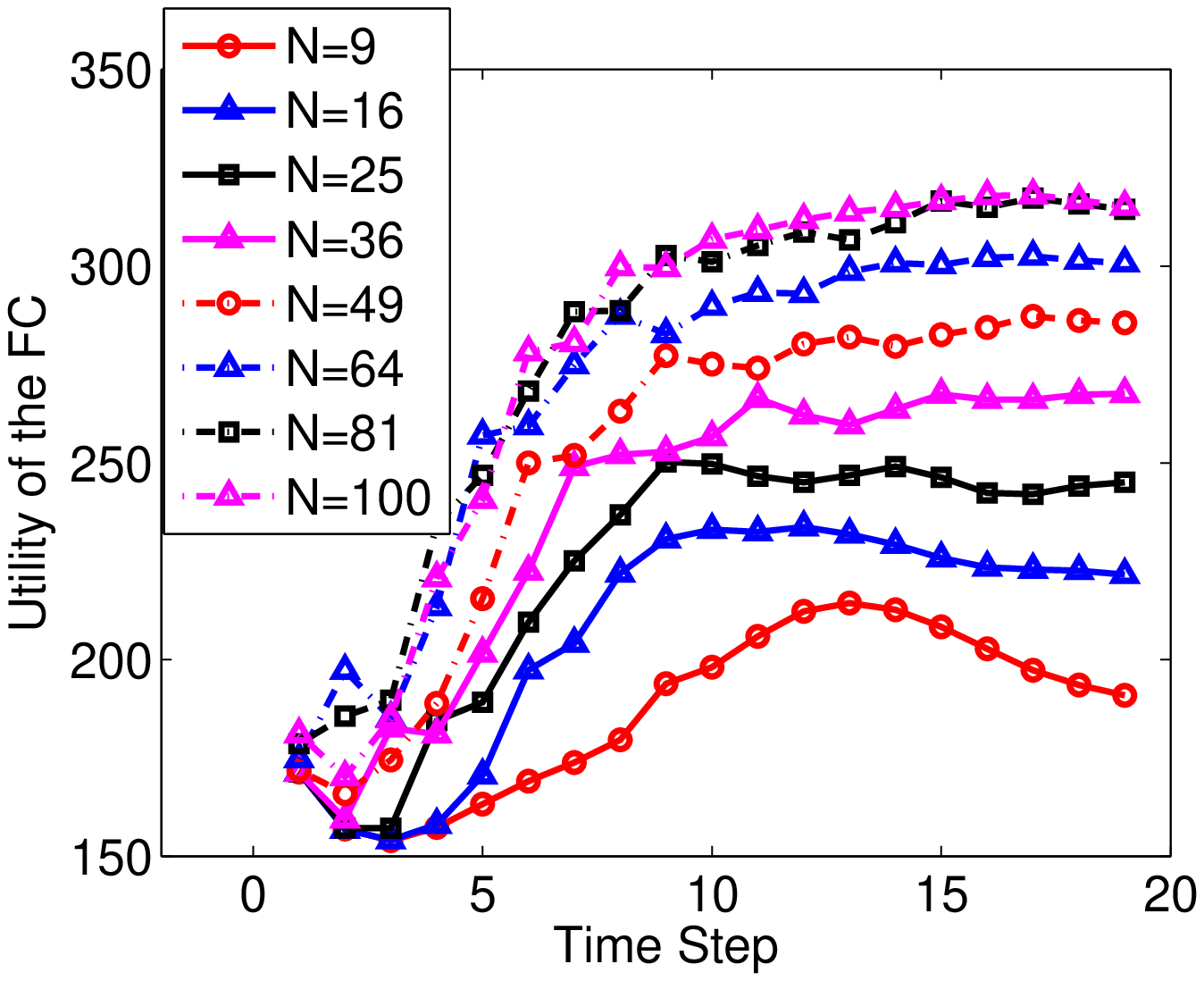}
\label{fig_U_N} }
\subfigure[]{
\includegraphics[%
  width=0.4\textwidth, height = 0.3\textwidth]{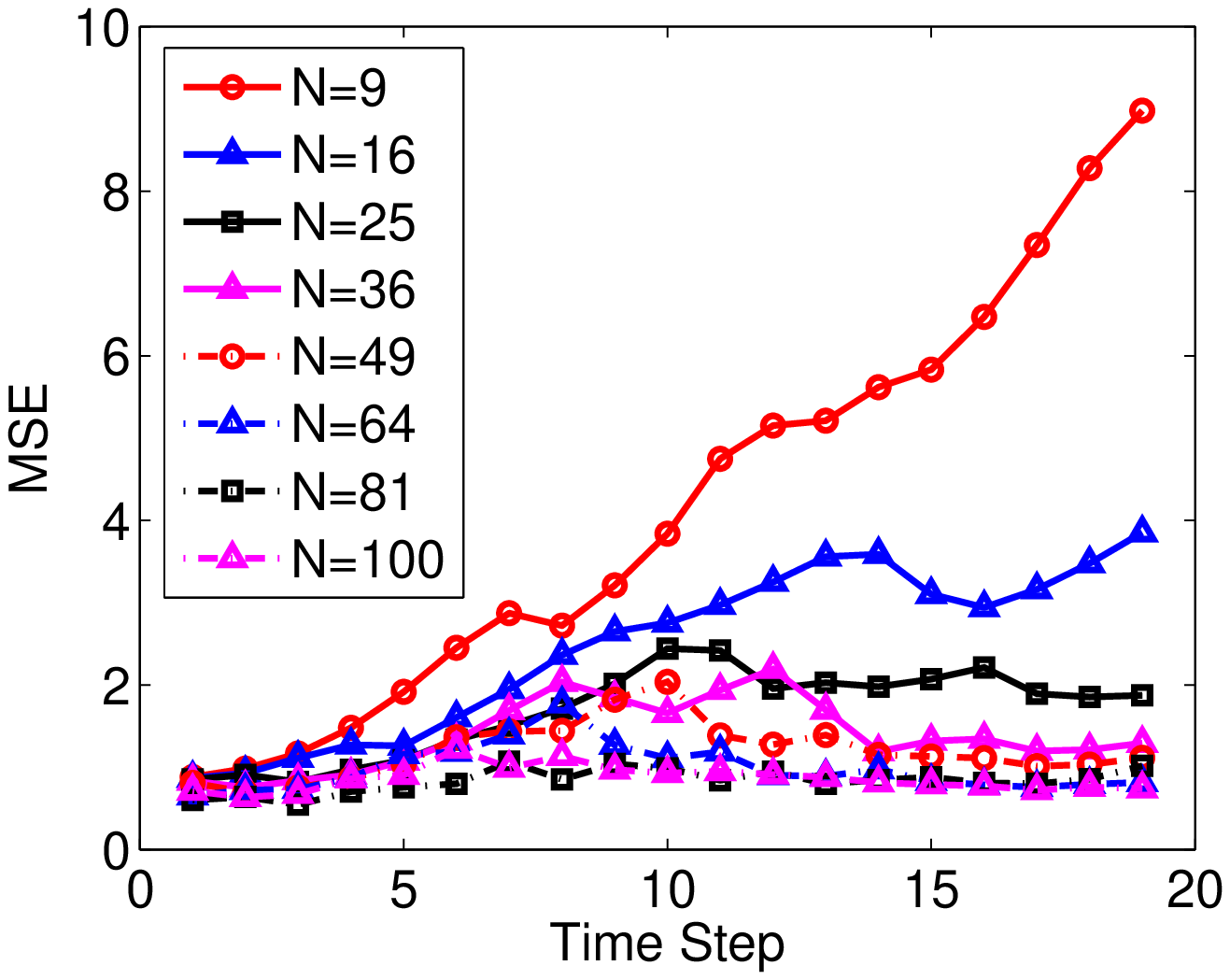}
\label{fig_MSE_N} }
\caption{With different number of sensors in the ROI (a) The utility of the FC. (b) The MSE at each time step.}
\label{u_mse_N}
\end{center}
\end{figure*}


In Fig. \ref{u_mse_N}, we study the utility of the FC (Fig. \ref{fig_U_N}) and the corresponding MSE (Fig. \ref{fig_MSE_N}) when there are different number of users in the network. The figures show that as the number of sensors in the WSN increases, the utility of the FC increases, and the corresponding MSE decreases. It is because as the sensor density in the ROI increases, the chances of the FC selecting more informative sensors that require less payment increase at each tracking step. In other words, competition among sensors increases as sensor density increases, thereby making sensors participate with lesser payments. Also, the FC's utility and the MSE saturate when the number of sensors in the ROI is large. This is because, as has been observed in economic theory, a large number of competitors in a market correspond to a scenario of perfect competition and result in the market prices to saturate. Note that when $N=9$ and 16, the utility of the FC decrease and the MSE diverges after a certain time. This is due to the fact that the number of sensors is not sufficient for accurate tracking over the large ROI.

Now we consider the case mentioned in Section \ref{Energy_Efficiency} where the sensors value their remaining energy. In \eqref{track_model}, we consider $\mathcal{D} = 1$ second and the observation length is 40 seconds.
The mean of the prior distribution is assumed to be $\mu_0 = [-10~-11~2~2]^T$ and the other parameters are kept the same.
The target moves back and forth between two different points. During the first and the third 10 second intervals, the target
moves as described by model \eqref{track_model} in the forward direction. At other times during the second and fourth 10 second intervals, the target moves in the reverse direction with all
other parameters fixed. For the function $g(e_{i,t-1})$ in \eqref{eq2_ee}, we take an example where the value estimates of the sensors increase
as their remaining energy decreases according to $g(e_{i,t-1}) = 1/(e_{i,t-1}/E_{i,0})^k$, where $E_{i,0}$ is the initial energy
of each sensor at the initial time step, and the power $k$ controls the increasing speed.
In Fig. \ref{fig_LT}, we show 
\begin{inparaenum}[\itshape a\upshape)]
\item the remaining number of active sensors in the WSN of the FIM based bit allocation algorithm in \cite{engin_bit_allocation},
\item our auction based bit allocation without residual energy consideration, and, 
\item when residual energy is considered with different exponent $k$.
\end{inparaenum} 
Note that in \cite{engin_bit_allocation}, the property of the determinant of the FIM brought the suboptimality of the 
approximate dynamic programming method. Here, to compare with our work, we employ the trace of the FIM as the bit allocation metric to get the optimal solutions using dynamic programming. For FIM based bit allocation algorithm and our algorithm without residual energy consideration, a specific bandwidth allocation maximizes the FC's utility for a given target location, resulting in the same set of sensors to be repeatedly selected  (as the target travels back to pre-visited locations) until the sensors die 
(sensors are defined to be dead when they run out of their energy). Thus, those sensors die earlier than the others and the number of active sensors decreases rapidly.
However, the increase of the value estimate based on residual energy prevents the more informative sensors from being selected repeatedly because
they become more expensive if they have already been selected earlier. In other words, on an average, sensors are allocated lesser number of bits as their residual energy decreases.  
Moreover, the larger the exponent $k$ is, the more the sensors value their remaining energy. 
We define the lifetime of the sensor network as the time step at which the network becomes non-functional (we say that the network is non-functional when a specified percentage $\alpha$ of the sensors die~\cite{Yunxia_lifetime}). For example, we assume $\alpha=0.6$, in the energy unaware case, the lifetime of the network is around 22. However, by our algorithm, the lifetime of the network gets extended to 30 when $k=3$, and even gets extended to the last time step when the tracking task ends with $k=15$ or $k=30$, i.e., the network keeps functional until the last tracking step.

\begin{figure}[htb]
\centering
  \includegraphics[width=.75\columnwidth]{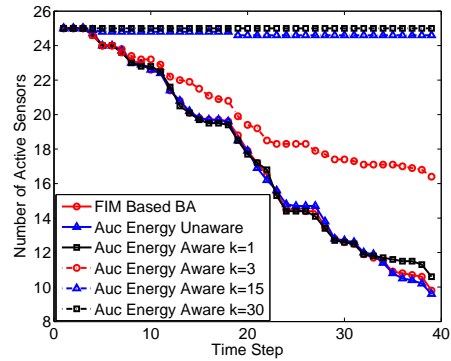}
\caption{The remaining number of active sensors in the WSN}
\protect \label{fig_LT}
\end{figure}

\begin{figure*}[tb]
\begin{center}
\subfigure[]{
\includegraphics[%
  width=0.4\textwidth, height = 0.3\textwidth]{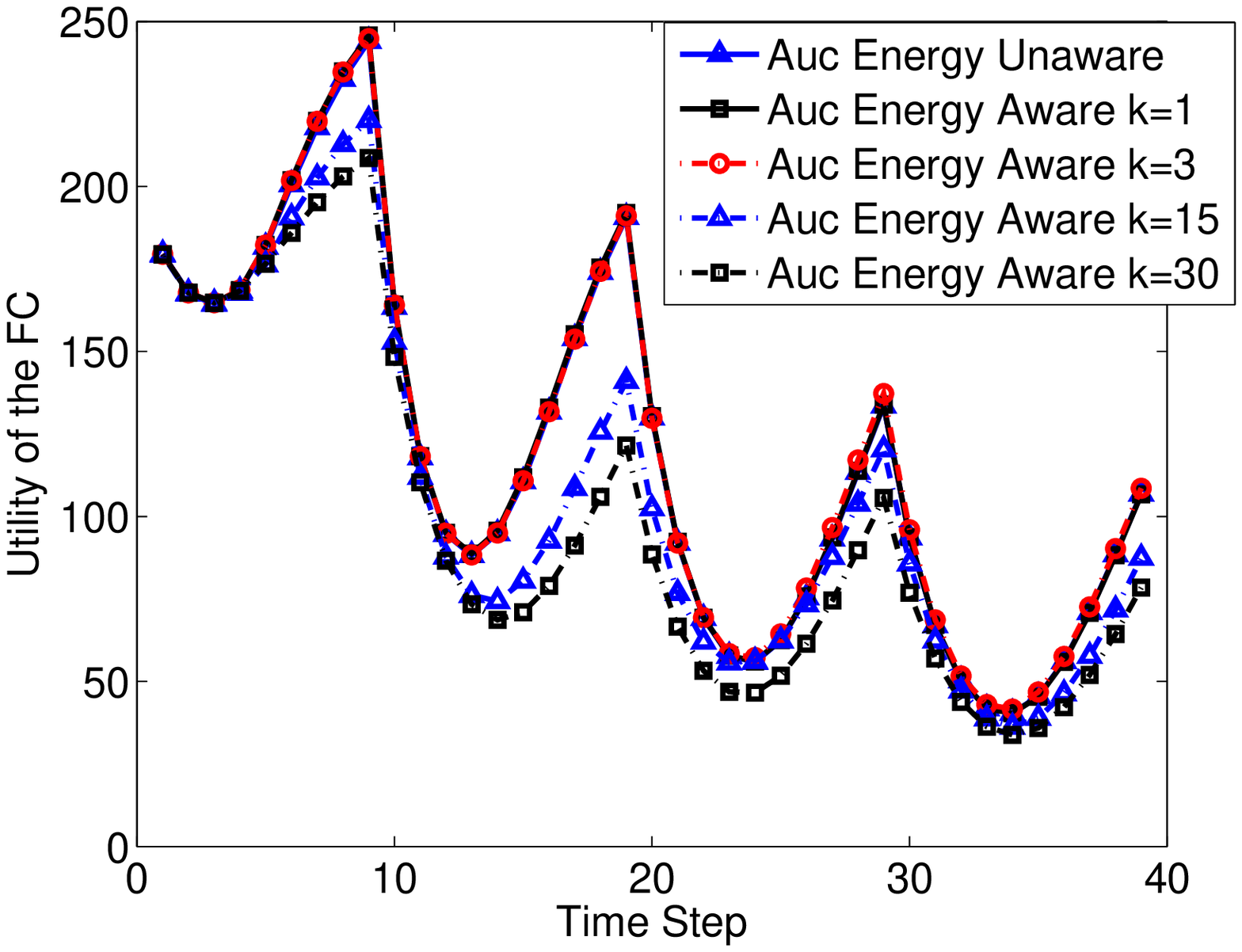}
\label{fig_Ufc} }
\subfigure[]{
\includegraphics[%
  width=0.4\textwidth, height = 0.3\textwidth]{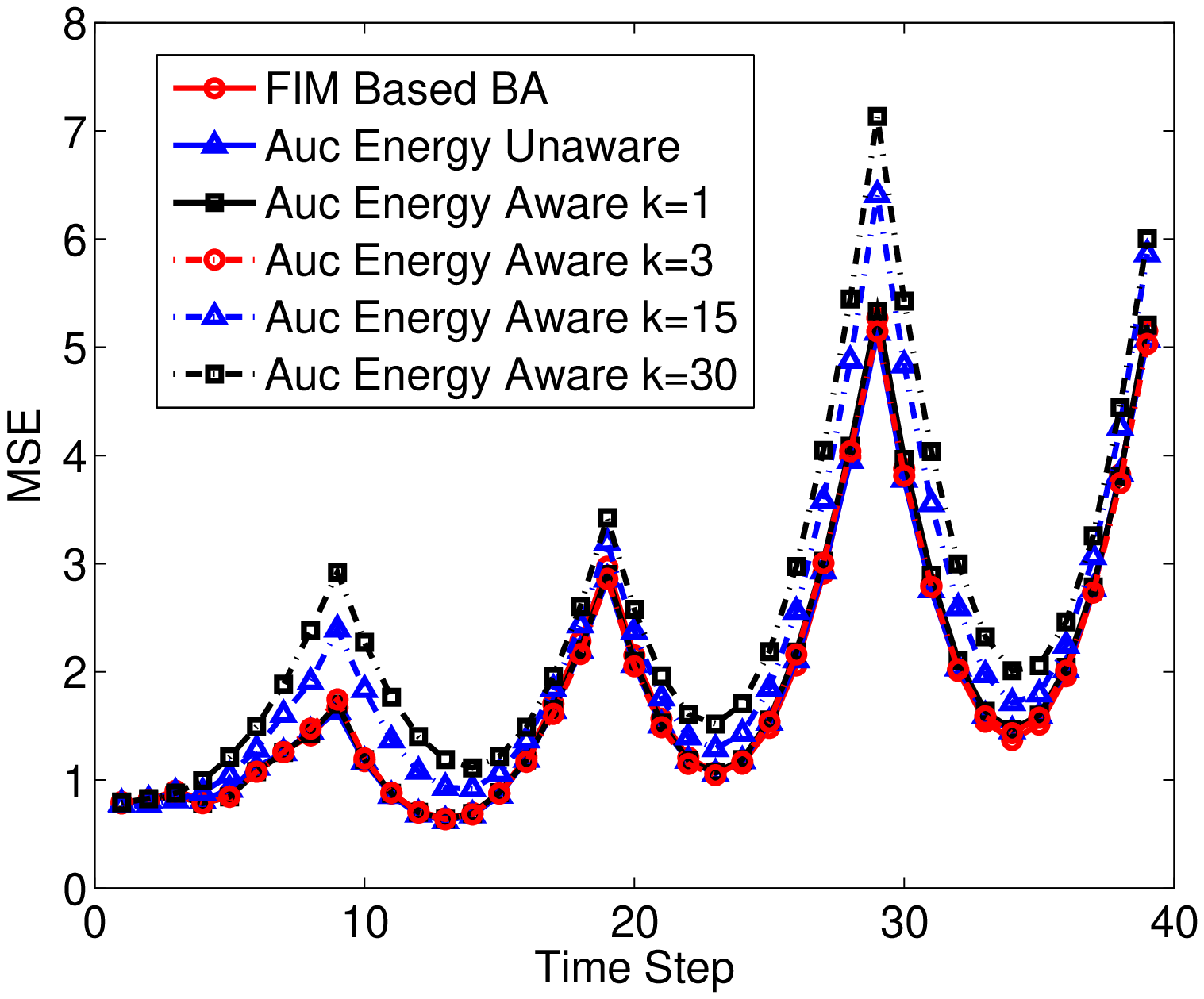}
\label{fig_MSE} }
\caption{With residual energy consideration (a) The utility of the FC. (b) The MSE at each time step.}
\label{u_mse_multiway}
\end{center}
\end{figure*}

Corresponding to Fig. \ref{fig_LT}, in Fig. \ref{u_mse_multiway}, we study the tradeoff of considering the function $g(e_{i,t-1}) = 1/(e_{i,t-1}/E_{i,0})^k$ in terms of the utility of the FC (Fig. \ref{fig_Ufc}) and MSE (Fig. \ref{fig_MSE}). 
As shown in Fig. \ref{fig_MSE}, the FIM based bit allocation algorithm gives the lowest tracking MSE because the sensors with highest Fisher information are always selected by the FC. For our algorithm without energy consideration and with residual energy considered as $k=1$ and $k=3$, the loss of the estimation error and the utility of
FC are very small. However, the loss increases 
when $k$ increases to $15$ and $30$. This is because when the sensors increase their valuations more aggressively,  
they become much more expensive after being selected for a few times. Then the FC, in order to maximize its utility, can only afford to select those cheaper (potentially non-informative) sensors and allocate bits to them.
In other words, depending on the characteristics of the energy concerns of the participating users, 
the tradeoff between 
the estimation performance and the lifetime of the sensor network is automatically achieved.

%

%
%

\section{Conclusion}

In this paper, we have designed a mechanism for the dynamic bandwidth allocation problem in the myopic target tracking problem by considering the sensors to be selfish and profit-motivated. To determine the distribution of the limited bandwidth and the pricing function for each sensor, the FC conducts an auction by soliciting bids from the sensors, which reflects how much they value their energy cost. Furthermore, our model guaranteed the rationality and truthfulness of the sensors. We implemented our model by formulating the optimization problem as a MCKP, which is solved by dynamic programming optimally. Also, we studied the fact that the trade-off between the utility of the FC and the lifetime of the sensor network can be achieved when the valuation of the sensors depend on their residual energy. In the future, we will study the mechanism design approach for the non-myopic target tracking problem.


\section*{Acknowledgment}
This work was supported by U.S. Air Force Office of Scientific Research
(AFOSR) under Grants FA9550-10-1-0458.





%

\bibliography{journal}
\bibliographystyle{IEEEtran}

%
%
%


\end{document}